%% file: main.tex
\documentclass[aps,prr,reprint,superscriptaddress,showpacs,showkeys,footinbib]{revtex4-2}
\usepackage{amsmath, amssymb, amsthm} 
\usepackage{bbm}
\usepackage{graphicx} 
\usepackage[colorlinks]{hyperref} 
\usepackage{physics} 
\usepackage{tikz}
\usetikzlibrary{quantikz2}
\usepackage{stmaryrd}  
\usepackage{bibunits} 
\defaultbibliographystyle{apsrev4-2}
\usepackage{siunitx}
\usepackage{appendix}
\usepackage{todonotes}
\usepackage{listings}
\DeclareMathOperator{\CardOp}{Card}
\newcommand{\Card}[1]{\CardOp\left(#1\right)}
\theoremstyle{definition}

\theoremstyle{plain}

\newtheorem*{lemma}{Lemma}
\newtheorem*{theorem*}{Theorem}

\def\equationautorefname~#1\null{Eq.\,(#1)\null}
\def\figureautorefname~#1\null{Fig.\,#1\null}
\def\sectionautorefname~#1\null{App.\,#1\null}

\DeclarePairedDelimiter\doublebrackets\llbracket\rrbracket

\newcommand{\bigO}[1]{\ensuremath{\mathop{}\mathopen{}O\mathopen{}\left(#1\right)}} 
\DeclarePairedDelimiter\ceil{\lceil}{\rceil}
\DeclarePairedDelimiter\floor{\lfloor}{\rfloor}

\begin{document}

\title{Network Requirements for Distributed Quantum Computation}

\author{Hugo Jacinto}
\affiliation{Alice \& Bob, 53 boulevard du Général Martial Valin, \num[detect-all]{75015} Paris, France}
\affiliation{Université Paris--Saclay, CNRS, CEA, Institut de Physique Théorique, \num[detect-all]{91191} Gif-sur-Yvette, France}
\email{hugo.jacinto@alice-bob.com}
\author{Élie Gouzien}
\affiliation{Alice \& Bob, 53 boulevard du Général Martial Valin, \num[detect-all]{75015} Paris, France}
\author{Nicolas Sangouard}
\affiliation{Université Paris--Saclay, CNRS, CEA, Institut de Physique Théorique, \num[detect-all]{91191} Gif-sur-Yvette, France}
\date{April 15, 2025}

\begin{abstract}
    Physical constraints and engineering challenges, including wafer dimensions, classical control cabling, and refrigeration volumes, impose significant limitations on  the scalability of quantum computing units.
    As a result, a modular quantum computing architecture, comprising small processors interconnected by quantum links, is emerging as a promising approach to fault-tolerant quantum computing.
    However, the requirements that the network must fulfill to enable distributed quantum computation remain largely unexplored.
    We consider an architecture tailored for qubits with nearest-neighbor physical connectivity, leveraging the surface code for error correction and enabling fault-tolerant operations through lattice surgery and magic state distillation.
    We propose measurement teleportation as a tool to extend lattice surgery techniques to qubits located on different computing units interconnected via Bell pairs.
    Through memory simulations, we build an error model for logical operations and deduce an end-to-end resource estimation of Shor's algorithm over a minimalist distributed architecture.
    Concretely, for a characteristic physical gate error rate of ${10}^{-3}$, a processor cycle time of $1$ microsecond, factoring a 2048-bit RSA integer is shown to be possible with 379 computing processors, each made with \num{89781} qubits, with negligible space and time overhead with respect to a monolithic approach without parallelization, if 70 Bell pairs are available per cycle time between each processor with a fidelity exceeding \SI{98.1}{\percent}.
    
\end{abstract}

\maketitle

\input{Introduction}
\medskip 
\noindent
\input{Definition_seam}

\medskip 
\noindent
\input{Logical_ansatz_simulation}
\medskip 
\noindent
\input{Minimalist_architecture}
\medskip 
\noindent
\input{Resource_estimation}
\section{Conclusion}
Through circuit level-simulations of memory experiments and given a minimalist architecture, we have identified network requirements that enable the execution of Shor's algorithm in a distributed platform using lattice surgery with little or no space-time overhead with respect to the monolithic approach.
Concretely, with a chain-like architecture made of 379 processors with \num{89781} qubits each, factoring 2048-RSA integers is shown possible in less than 22 days for a physical gate error rate of ${10}^{-3}$ if 
 $2d=70$ Bell pairs are available per cycle time (\SI{1}{\micro\second}) between each nearest-neighbor processor, with a fidelity exceeding \SI{98.1}{\percent}.
The simultaneous requirements on fidelity and rate have not yet been achieved experimentally~\cite{remote_cnot_ibm, Yan_2022, Storz2023, Storz_2025}, but our findings provide valuable guidance on how network performance must improve in order to scale up quantum computing.  
As a perspective, if inter-cryostat connections do not reach the required fidelity, one can envision employing entanglement distillation to target fidelities of about 99\% starting from initial values around 90\%. Achieving this would likely require only a small number of physical qubits compared to the size of each processor.

\paragraph*{Note added} While we were finalizing this manuscript, the proposed method for distributing rotated surface code patches using similar techniques was independently reported in Ref.~\cite{shalby2025optimizednoiseresilientsurfacecode}.

\begin{acknowledgments}
 We thank Jérémie Guillaud, Christophe Vuillot and Diego Ruiz for insightful discussions.
This work was partially supported by the French National program Programme d’investissement d’avenir, IRT Nanoelec, with the reference ANR-10-AIRT-05.
This research is partially funded by the French State through France 2030.

\end{acknowledgments}

\bibliography{sample}
\clearpage

\subsection*{Supplemental Material}
\twocolumngrid
\appendix  
\setcounter{secnumdepth}{2}
\input{appendix_teleported_measurement}
\input{appendix_seam}
\input{appendix_ansatz}
\input{appendix_numerical_methods}

\input{appendix_architecture}

\input{appendix_lattice_surgery}
\input{appendix_magic_state_v2}
\input{appendix_shor}

\end{document}

%% file: Introduction.tex
\section{Introduction}

Quantum algorithms solving practical problems typically require between ${10}^7$ and ${10}^{11}$~\cite{Beverland2022, Dalzell2023, scholten2024assessingbenefitsrisksquantum} elementary gates to be executed.
Available quantum processors, with typical error rates of ${10}^{-3}$ per operation~\cite{Quantinuum2023, Acharya2023, Kim2023}, must incorporate fault-tolerant quantum error correction to implement such sequences of operations reliably~\cite{Campbell2017}.
These methods come however, with a significant qubit overhead.
Considering the methods of surface codes, lattice surgery and magic state distillation~\cite{Litinski_2019} and an error rate of ${10}^{-3}$, it is estimated that algorithms requiring ten thousand logical qubits ultimately demand tens of millions of physical qubits~\cite{Gidney_2021}.
Although alternative quantum error correction protocols have been recently proposed~\cite{Bravyi2024, gu2024optimizingquantumerrorcorrection, gidney2023yokedsurfacecodes} and various physical systems are envisioned to build a fault-tolerant quantum computer, scaling in a monolithic approach will be a daunting challenge for all of them.


\smallskip 
\noindent
A solution to this scaling issue consists in using a network of small scale processors.
Achieving true scalability then reduces to the design of a processing unit with a small number of logical qubits and an interface that can be used to interconnect the different units by means of network links.
Although algorithms can be split into pieces that can be connected classically by mimicking quantum operations through classical connections~\cite{CarreraVazquez2024}, it comes with an exponential cost~\cite{jing2024circuitknittingfacesexponential,Piveteau_2024} and the generic approach uses quantum links.
The basic idea, already involving experimental efforts~\cite{remote_cnot_ibm,Storz2023, Yan_2022, Storz_2025}, is to establish Bell pairs between the processing units so that qubits can be transferred or gates can be operated between processing units by means of teleportation~\cite{Caleffi_2024,distributed_qec_cosmic_ray}.
Because the links are noisy and lossy, entanglement distillation protocols~\cite{Dur_2003, Campbell_2007, Krastanov_2019} were initially thought to be necessary to convert multiple low quality Bell pairs into a small number of high fidelity Bell pairs and guarantee high quality operations between processors.
These protocols~\cite{Bennett_1996,Deutsch_1996} are either simple but increase the memory error per cycle or are more complex but increase the error correction cycle time~\cite{Liang_2007,Fujii_2012,Nickerson_2013,Nickerson_2014}.
An alternative is to distribute the computation at the level of the error correction.
Recent results, showed that logical qubits encoded in unrotated surface code patches in distinct processing units can be fault-tolerantly connected with noise threshold along their shared interface reaching up to \SI{10}{\percent}~\cite{ramette2023faulttolerantconnectionerrorcorrectedqubits, Sinclair_2024}.
Although these results are encouraging, the central questions about how many Bell pairs are required and how faithful they need to be to execute a quantum algorithm in a distributed way with comparable space and time resources with respect to the monolithic approach remain open.

\smallskip 
\noindent
We here consider a distributed architecture where each processing unit consists of a 2D grid of physical qubits.
The layout involves columns of logical data qubits encoded in rotated surface code patches~\cite{fowler2019lowoverheadquantumcomputation} that are separated by physical routing qubits to achieve all-to-all logical connectivity through lattice surgery.
Specifically, we show how logical multi-qubit gates operating on qubits located on separate units can be implemented fault-tolerantly by combining lattice surgery and a technique that we introduce as ``measurement teleportation'': using a Bell pair to measure a non-local operator without further interaction between remote subsystems. 
From circuit-level simulations of logical memory experiments, we derive an error model for these multi-qubit gates operating remotely.
We then consider an execution of Shor's factorization algorithm on this distributed architecture and provide a detailed estimation of the runtime, qubit and Bell pair numbers needed to factor large RSA integers.
Using plausible physical assumptions for processors based on superconducting qubits, namely, a characteristic physical gate error rate of ${10}^{-3}$ and a surface code cycle time of 1 microsecond, we show that factoring a 2048-bit RSA integer is feasible with 379 processors, each comprising \num{89781} physical qubits.
This can be achieved with negligible overhead compared to using the same layout in a monolithic approach, provided that 70 Bell pairs are distributed per cycle time between each processor with a fidelity exceeding \SI{98.1}{\percent}.

\medskip 
\noindent
\section{Circuit-level simulation of the rotated surface code}
We first focus on the rotated version of a distance $d$ surface code~\cite{orourke2024comparepairrotatedvs}, namely each logical qubit is encoded in a $d\times d$ square patch, where $d^2$ physical qubits sit on the vertices.
$d^2-1$ stabilizer generators, which are given by the 4-qubit and 2-qubit tensor products of either $X$ or $Z$ Pauli operators, are on the plaquette.
The logical Pauli operators are strings of $X$ and $Z$ Pauli operators along any vertical and horizontal directions of the patch, respectively.
Time is discretized into rounds, with each round corresponding to the measurement of all stabilizers, and the resulting outcomes referred to as a syndrome.
In order to deal with faulty measurements, rounds are repeated $O(d)$ times~\cite{Dennis_2002} before being analyzed through classical decoding to identify appropriate correction.

\smallskip 
\noindent
To evaluate accurately the performance of the rotated surface code, we perform memory experiment simulations at the circuit level.
Concretely, the results of the $Z$ stabilizers are obtained by means of CNOT gates, each controlled by the data qubits and targeting an ancilla qubit initially prepared in the $\ket{0}$ state and finally measured in the $Z$ basis.
The CNOTs are reversed for the $X$ stabilizer measurements, with the ancilla qubit prepared initially in the $\ket{+}$ state and finally measured according to the $X$ basis.
The order of the CNOT gates is chosen so that 4 time steps are enough to extract the results of all stabilizers while avoiding hook errors and preserving the code distance~\cite{Dennis_2002,Fowler_2012, Yu_2014}.
We use the fully depolarizing noise model characterized by the parameter $p$, which we refer to as the physical error rate, see \autoref{annexe:noise model} for details.
For the decoding, we choose a minimum-weight-perfect-matching decoder~\cite{higgott2021pymatchingpythonpackagedecoding}.
We consider patches encoding either a logical $\ket{+}$ or $\ket{0}$ and we compute the logical error rate separately.
In both cases, we retrieve a sub-threshold behavior where errors are exponentially suppressed~\cite{Dennis_2002}.
Specifically, the per-round probability to have a $X$ logical error on $d\times d$ square patch is extracted from a $\ket{0}$ memory experiment with 
\begin{equation}\label{eq:log error rate model}
    P_L^X \approx \alpha {\left(\frac{p}{p_{\text{th}}}\right)}^{\frac{d+1}{2}}
\end{equation}
where $p_{\text{th}} = (7.43\pm 0.08) \times {10}^{-3}$ and $\alpha = (5.0 \pm 0.2) \times {10}^{-2}$, see \autoref{annexe: fit methodology} for details on the fitting methodology.
Similarly, $P_L^Z$ associated to $\ket{+}$ is the same, up to uncertainty on the fitted parameters.
Conservatively, we attribute a logical error rate per round $P_L =P_L^X +P_L^Z$ to the storage of an arbitrary state.

%% file: Definition_seam.tex
\section{Remote stabilizer measurement circuit via measurement teleportation}\label{section frontier}
%
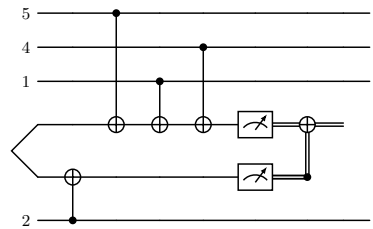
\begin{figure}
     \centering
        \scalebox{0.7}{
        \begin{quantikz}
        \lstick{$2$}&\ctrl{3}  &          &          &         &       &&&\\
        \lstick{$4$}&          &          &          &\ctrl{2} &       &&&\\
        \lstick{$1$}&          &          &\ctrl{1}  &         &       &&&\\
        \makeebit{} & \targ{}  &          & \targ{}  &\targ{}  &\meter{}&\targ{}\setwiretype{c}&\\
                    &          & \targ{}  &          &         &\meter{}&\cwbend{-1}\setwiretype{c}\\
        \lstick{$5$}&          &\ctrl{-1} &          &         &       &&&
        \end{quantikz}
        }
    \caption{
    Teleported $Z$-stabilizer measurement circuit.
    The first four qubits are located on the same computing unit, while the last two are situated on a separate unit.
    The connection between the fourth and fifth qubits indicates that they are prepared in a Bell state.
    The data qubit numbering follows that of \autoref{fig:seam}.}
    \label{fig:syndrome extraction}
\end{figure}
A key challenge when physical qubits are distributed across remote computing units is the efficient extraction of stabilizer measurement results, which is typically achieved through teleportation.
This involves either teleporting a quantum state from one unit to another or teleporting a CNOT gate between an ancilla qubit and a data qubit~\cite{Caleffi_2024}.
We propose a strategy utilizing two ancilla qubits, each located in a separate unit and initially prepared in a Bell state, $1/\sqrt{2}(\ket{00}+\ket{11})$.
For $Z$-stabilizer, local CNOT gates are then applied between the data qubit and the ancilla qubit within the same unit, with the ancilla serving as the target.
The stabilizer measurement result is extracted by measuring the ancilla qubits in the $Z$ basis and combining the outcomes with a modulo-2 addition (XOR), see \autoref{Annexe:Teleported measurement}.
To extract the $X$-stabilizer result, the roles of control and target in the CNOT gates are reversed, followed by $X$-basis measurements of the ancilla qubits.
In summary, a Bell pair enables a joint measurement between distant qubits without direct interaction.
The outcome is known only after the classical communication of measurement results.
This process is referred to as remote stabilizer measurement via measurement teleportation.

\begin{figure}
    \centering
    \includegraphics[width=\linewidth]{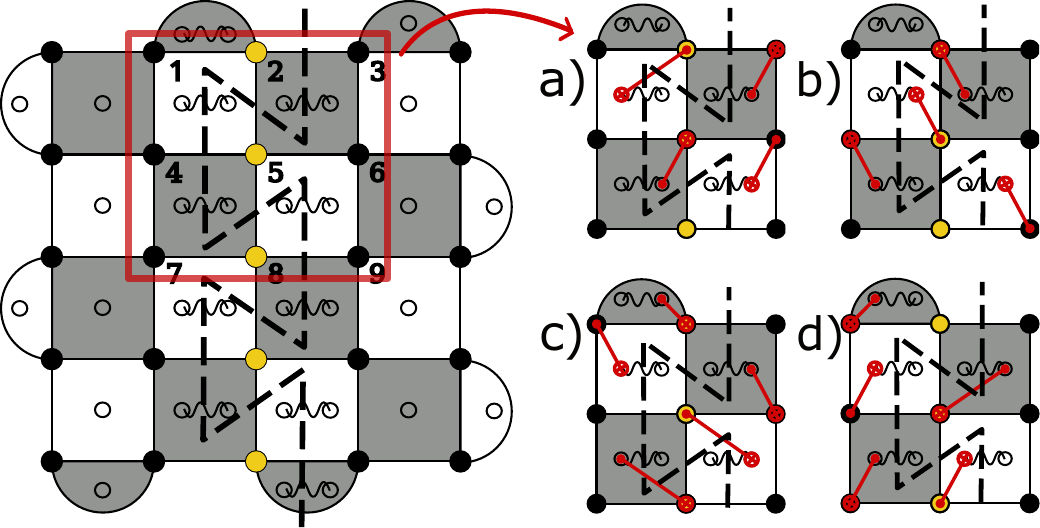}
    \caption{Distributed rotated surface code definition.
    Filled dots represent data qubits, and empty dots represent ancillary qubits.
    Yellow filed-dot represent data qubits belonging to the ``seam'', the line of data qubits between the two columns of Bell pair ancillae.
    Empty dots connected by a sinusoidal line correspond to ancillary qubits prepared in a Bell state.
    The dark grey plaquettes indicates the X-stabilizers over the data qubits at their edges, while the white plaquettes represent $Z$-stabilizers.
    The dashed line represents the physical separation between the two computing units.
    The data qubit numbering help defining the order with which the CNOTs of the extraction circuits are applied, see \autoref{fig:syndrome extraction}.
    All stabilizer can be measured in the four times steps shown in figures labels a),b),c) and d).}
    \label{fig:seam}
\end{figure}

\smallskip 
\noindent
The next question is how many data qubits need to be connected to each ancilla in the Bell pair? At first glance, one might assume that a single vertical line of ancillary qubits could be replaced with Bell pairs.
This would mean connecting two data qubits to each ancilla, requiring only $d_X$ Bell pairs, where $d_X$ is the distance with respect to $X$ errors.
However, such a configuration suffers from hook errors, namely one $X$ error on a Bell pair used to measure a $X$-stabilizer would translate into two $X$-errors on two data qubits located on the same vertical line of the surface code patch.
This effectively halves the 
$X$-distance, as detailed in \autoref{naive seam}.
The stabilizer measurement circuit presented in \autoref{fig:syndrome extraction} provides a solution.
The idea is to connect three qubits to the same ancilla in the right order so that one $X$-error on the Bell pair would translate in the worst case into two errors on qubits belonging to the same horizontal line of the surface code patch, hence preserving the effective distance.
This effectively takes $2d_X$ Bell pairs.
Note that the implementation of the CNOTs can still be organized so that the syndrome can be extracted in four time steps, see \autoref{fig:seam}.

%% file: Logical_ansatz_simulation.tex
\section{Circuit-level simulation of the distributed rotated surface code}\label{simulation results}
With the syndrome extraction circuits specified, the goal is now to evaluate the performance of the distributed rotated surface code through memory experiment simulations.
We consider square patches with various distances and two columns of Bell pair ancillae aligned with $X_L$ logical operators, see \autoref{fig:seam}.
The noise model remains the same as before for idle qubits, measurements, single-qubit gates, and CNOT operations.
For the noise affecting the ancilla qubits, we consider two-qubit depolarizing noise on the Bell pairs, characterized by an occurrence probability $p_{\text{Bell}}$.
This corresponds to a noisy Bell pair fidelity given by $\mathcal{F}\left(\rho, \ketbra{\Phi^+}\right)=1-\frac{4}{5}p_{\text{Bell}}$.
As detailed in \autoref{Annexe noise Bell}, this preparation noise only propagates to the data qubits located at the ``seam'', i.e.\@ between the two lines of teleported stabilizers.
Decoding is still carried out using a minimum-weight perfect matching decoder.
As before, logical error rates $P_L^X$ and $P_L^Z$ are characterized by storing logical $\ket{0}$ and $\ket{+}$ respectively.
For large distances, Bell pair infidelity has a negligible effect on the $Z$-logical error rate, as errors on the seam only affect one data qubit of any logical Pauli $Z$ operator.
Therefore, in this section we restrict our analysis to $X$-logical errors.


\smallskip
\noindent
Rough combinatorics to count error chains in the $X$-error matching graph of a rotated surface code under a phenomenological noise model motivates the following ansatz for the $X$-logical error rate, see \autoref{annexe: ansatz}
\begin{multline}\label{eq:split log er rate model}
    P_{L}^X \approx
    \alpha_1{\left(\frac{p_{\text{Bell}}}{p_{\text{Bell}}^*}\right)}^{\frac{d+1}{2}} +  \alpha_2{\left(\frac{p}{p^*}\right)}^{\frac{d+1}{2}} \\ 
    + \alpha_3 \sum_{1\leq i\leq d}{\left(\frac{p_{\text{Bell}}}{p_{\text{Bell}}^{**}}\right)}^{\frac{i}{2}} {\left( \frac{p}{p^*}\right)}^{\frac{d+1-i}{2}}.
\end{multline}
The parameter $p$ represents the error rate for all operations except the Bell pair preparation that is characterized by $p_{\text{Bell}}$, see \autoref{annexe:noise model}.
Threshold quantities are indicated with a superscript $^*$.
Specifically, we define a pseudo-threshold as $p_{\text{Bell}}^{**}=p_{\text{Bell}}^*{\left(1+\frac{\alpha_c}{1-\sqrt{p/p^*}}\right)}^{-2}$.
$\alpha_1$, $\alpha_2$, $\alpha_3$ and $\alpha_c$ are constants that need to be determined through fitting.
Despite coming from a simplified noise model, this ansatz appears to capture most of the underlying noise diffusion processes.
Transitioning to a circuit-level noise model induces additional diagonal edges in the matching graph due to CNOT noise diffusion.
The error chains counting leading to \autoref{eq:split log er rate model} still stand on this knotty graph but are characterized by different combinatorial constants resulting in modified thresholds~\cite{Dennis_2002}.

\begin{figure}[ht!]
    \centering
    \includegraphics[width=0.9\linewidth]{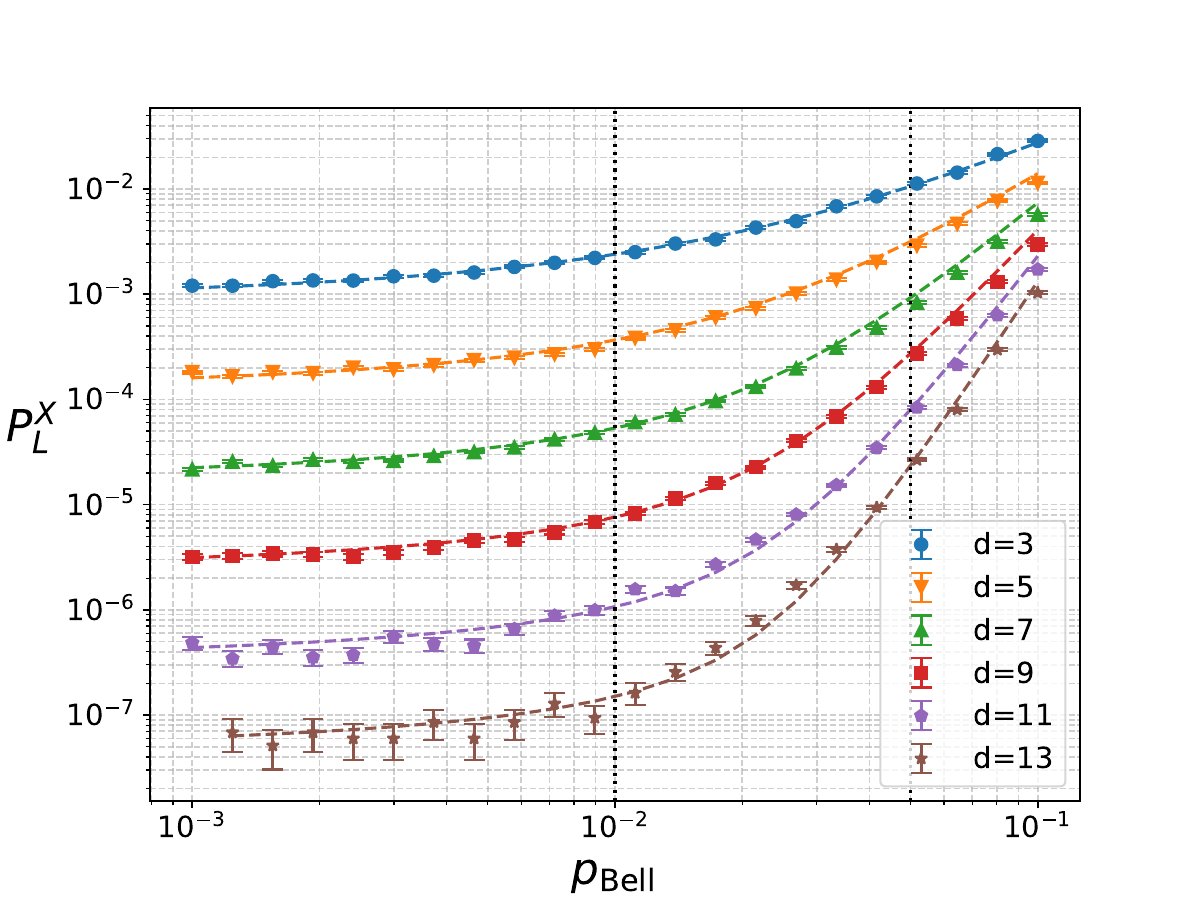}
    \caption{$X$-logical error rate per round $P_L^X$ as a function of the physical error rate on the Bell pairs $p_{\text{Bell}}$, under the assumption that every operation is subject to depolarizing noise with a rate  $p={10}^{-3}$.
    The results are presented for various code distances, as indicated in the inset.
    The dashed curves are the logical error model deduced from a collective fit of the ansatz of \autoref{eq:split log er rate model}.
    The two vertical dotted lines delimits the interval in which we use this model in the resource estimation.}\label{fig:p fixed slice}
\end{figure}

\smallskip
\noindent
\autoref{eq:split log er rate model} is used to reproduce the memory simulation results.
The fit results, covering a set of points in the $(p_{\text{Bell}}, p, P_L^X)$ space, lead to $\alpha_1=0.98\pm0.02$, $\alpha_2=0.045\pm0.002$, $\alpha_3=0.053\pm0.002$, $\alpha_c=0.21\pm0.02$, $p^*=(7.18\pm0.04)\times {10}^{-3}$ and $p_{\text{Bell}}^*=0.298\pm0.011$, see \autoref{annexe: fit methodology} for details of the fitting procedure.
\autoref{fig:p fixed slice} shows the result of this fit in the region of interest for our study corresponding to the slice with $p={10}^{-3}$.
Interestingly, the tolerance to Bell pair infidelity is remarkably high, with a threshold $p_{\text{Bell}}^*$ close to \SI{30}{\percent}.
This can be understood by focusing on errors arising during Bell pair preparation, that is considering the case where $p=0$.
As previously discussed, $X$-errors occurring on the Bell pairs propagate to the data qubits on the seam when the Bell pair is used to measure an $X$-stabilizer.
Similarly, they induce measurement errors when the Bell pair is used to measure a $Z$-stabilizer.
As a result, error correction in this scenario reduces to the task of a repetition code handling phenomenological noise associated with an error rate $8p_{\text{Bell}}/15$, see \autoref{Annexe noise Bell} for additional details.

%% file: Minimalist_architecture.tex
\section{Layout of a simple  distributed architecture}\label{sec:layout}
The layout is designed for a fault-tolerant implementation of gates on the rotated surface code leveraging lattice surgery~\cite{Horsman_2012} and magic state injection.
Each processor stores logical qubits into patches of distance $d$, grouped in two columns and interleaved with routing qubits within each computing unit, see~\autoref{fig: layout}.
The processors are arranged in a one-dimensional chain, successively interconnected by $2d$ Bell pairs.
This arrangement ensures protection against $\lfloor (d-1)/2 \rfloor$ errors in a patch resulting from a lattice surgery merging of any logical qubits.
This layout is particularly simple, but executing multiple lattice surgery operations in parallel is hardly possible since only one routing patch at a time can be shared between the units at the two extremities of the chain, see \autoref{Annexe lattice surgery} for further details.
Additionally, we do not consider a dynamical layout where logical qubits are relocated during the algorithm execution as done in~\cite{Gidney_2021}, hence enhancing parallelization while reducing the total number of routing qubits.
Nevertheless, our approach provides valuable insights into the space and time overhead of the distributed setting in the presence of Bell pair errors.

\begin{figure*}
    \centering
    \includegraphics[width=0.8\linewidth]{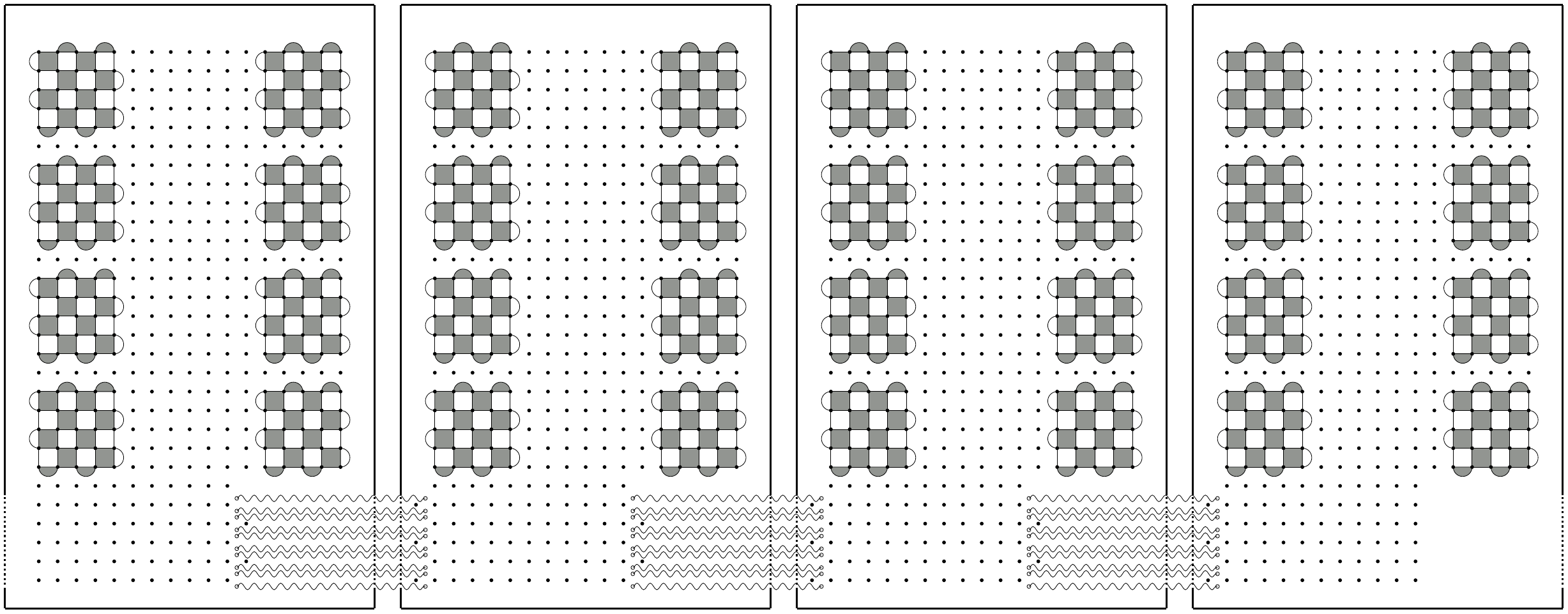}
    \caption{Chain architecture of processors with the compact layout, here represented with 4 processors.
    Dots are data qubits. For readability, ancillae qubits have been omitted except for Bell pairs represented by two circle linked by a sinusoidal line.
    The different patches exhibit the location of logical qubits in the layout.
    The other data qubits are routing qubits used to perform lattice surgery between the represented logical qubits.}
    \label{fig: layout}
\end{figure*}

%% file: Resource_estimation.tex
\section{Resource estimation to factor 2048-bit RSA integers}
We here report on a comparative resource estimation to factor a 2048-bit RSA integer using Ekerå and Håstad~\cite{Ekera_2017} version of Shor's algorithm~\cite{Shor} with a monolithic and a distributed architecture.
We begin by decomposing Shor's algorithm into a sequence of subroutines outlined in~\autoref{End_matter:resource}.
The number of gates in each subroutine is deduced from their implementations using a Toffoli-based gate set. 
This provides the total number of logical qubits and gates required for factoring a 2048-bit RSA integer.
Using explicit constructions of CNOT, C$X...X$ (multiple NOT gates controlled by one qubit) and Toffoli gates at the logical level (see \autoref{Annexe lattice surgery} and \autoref{Annexe: Distillation}), we then estimate their error rates and the time it takes to execute them as a function of the code distance.
For each code distance, the algorithm execution time is deduced by dividing the sum of the implementation time of all the operations (gates, measurements, and initializations) by the overall success probability.
The space cost is obtained from the number of physical qubits required by the layout to host the required number of logical qubits, magic state factories and routing areas for a given number of computing units.
We choose to take the smallest possible processors while ensuring that a magic states factory would fit inside, which also fixes the number of processors in the architecture.
The time and space overheads are deduced by choosing the distance, the subroutines parameters (detailed in \autoref{End_matter:resource}) and the magic state factory minimizing the volume.

\medskip 
\noindent
For CNOT gates, we consider an approach with an auxiliary qubit and two joint measurements performed via lattice surgery.
Since each joint measurement takes $\sim d$ rounds, the CNOT gate duration via lattice surgery is a priori $2dt_c$~\cite{Litinski_2019,fowler2019lowoverheadquantumcomputation}, where $t_c$ is the cycle time, namely the duration of one syndrome extraction cycle, which we assume to be $\SI{1}{\micro\second}$.
However, for some specific logical qubit locations, it happens that accessing logical operators for the second joint measurement requires an additional $d$ cycle to move patches.
Therefore, we have made the conservative assumption of attributing a time cost of $4dt_c$ to the CNOT gate, see \autoref{Annexe lattice surgery}.
The probability of an error-free implementation of such a CNOT is evaluated by considering that neither the control, nor the target, nor the other idling data qubits should experience any errors over $4d$ rounds while the error rate per round is given by \autoref{eq:log error rate model}.
If a CNOT is implemented between qubits belonging to different computing units, an error-free implementation also imposes that no error occurs on the auxiliary-control or auxiliary-target merged patches during $d$ rounds while the error rate per round, of the form \autoref{eq:split log er rate model}, depends on the number of seams separating the two computing units.
Note that both \autoref{eq:log error rate model} and \autoref{eq:split log er rate model} are models relying on memory experiment simulations, so that time-like errors effects are not captured.
We always consider the worst-case scenario in which the routing patch crosses the entire chain of computing units, 
see \autoref{Annexe lattice surgery} for the details on the noise model of CNOT gates and extension to C$X...X$.

\medskip 
\noindent
For the Toffoli gate, we use a two-level distillation factory combining protocols consuming 15 noisy $T$ states to produce one cleaner $T$ state, that is then used as input of a 8 $T$ to 1 CC$Z$ protocol.
The space cost of this factory is evaluated for different inner distances, see~\cite{Litinski_2019_magic} and  \autoref{annexe sec:Magic state factories}.
We use two factories, so that CC$Z$ states can be both injected and prepared in parallel as the injection time is longer than the duration of the magic state preparation.
This choice is further detailed in \autoref{annexe sec: Toffoli error model} along with the logical error model of the Toffoli gate, mainly based on the CNOT logical error model.

\medskip 
\noindent
The result of our resource estimation is presented in \autoref{tab:resource}, see also \autoref{Reading guide table}.
For an error rate $p={10}^{-3}$, our reference monolithic approach requires over 32 million qubits to factor a 2048-bit RSA integer, with a computation time exceeding 20 days.
Resource and runtime are essentially unchanged when the algorithm is executed in a distributed way without error on the Bell pairs.
The slight additional space cost is due to the presence of additional routing qubits in the distributed architecture, see \autoref{Appendix: layout}.
The optimal distance remains unchanged for up to $p_{\text{Bell}} = \SI{2,42}{\percent}$ corresponding to $\mathcal{F}\left(\rho, \ketbra{\Phi^+}\right) \geq \SI{98.1}{\percent}$.
However, since the time cost is defined as the algorithm's expected runtime toward its success, there is a slight overhead due to the increased failure probability of each logical CNOT gate in the circuit.
The increase of qubit cost for $p_{\text{Bell}}$ around $\SI{3}{\percent}$ is due to a change of distance, leading to each logical qubit requiring more physical qubits.

%% file: appendix_teleported_measurement.tex
\section{Measurement teleportation}\label{Annexe:Teleported measurement}

State teleportation as well as gate teleportation now stands as standard tools in quantum computing and communication.
Here we introduce the concept of measurement teleportation.
The principle is to use a Bell pair and classical communication to directly perform a non-local measurement.
It is a particularly useful tool in the context of quantum error correction to distribute computation over several QPUs.
Indeed, it stands as a natural technique to measure the stabilizer of some error correcting code standing on several processors.
Note that the idea to use GHZ state for syndrome extraction is common in the literature~\cite{DiVincenzo_2007_GHZ_syndrome}.
Indeed, $n$-qubits GHZ state has been originally used for syndrome extraction of $n$-qubits stabilizers to ensure fault tolerance.
Similarly to this idea we propose to measure $n$-qubits stabilizers split into two subsystems with $2$-qubits GHZ state, namely Bell pairs.

In teleported measurement, a Bell pair is shared between two subsystems and used as a single ancilla system to extract the measurement result.
Without entanglement between the two qubits of the ancilla system, the measurement would reveal too much information about the system by effectively measuring smaller operators than intended.
Note that a non-local measurement can be achieved through state teleportation, by teleportation of an ancilla qubit, or gate teleportation, by teleportation of two-qubit gates.
However, using measurement teleportation give shorter circuits\footnote{For instance, the 4-qubit parity $ZZZZ$ implemented in \autoref{fig:teleported_measurement} takes two time steps instead of four when a single qubit ancilla is available.}.
\subsection{Small example on $ZZ$ measurement}
Consider a two-qubit $ZZ$ measurement circuit through a single ancilla:
\begin{center}
    \centering
    \scalebox{0.8}{
    \begin{quantikz}
        \lstick{$1$}      &        & \ctrl{2}&          \\
        \lstick{$2$}      &\ctrl{1}&         &          \\
        \lstick{$\ket{0}$}&\targ{} &\targ{}  &\meter{}
    \end{quantikz}
    }
\end{center}
If qubit $1$ is on a first subsystem and $2$ on another, measuring the operator with two independent ancillae would reveal not only $Z_{1}Z_2$ but also $Z_1$ and $Z_2$ which will accidentally over-project the qubits: 
\begin{center}
    \scalebox{0.8}{
    \begin{quantikz}
        \lstick[2]{\sc first subsystem}&\wireoverride{n}\lstick{$1$}  &\ctrl{1} & \\
        &\wireoverride{n}\lstick{$\ket{0}$}& \targ{}   &\meter{}\\
        \lstick[2]{\sc second subsystem}&\wireoverride{n}\lstick{$\ket{0}$}&\targ{}       & \meter{}\\
        &\wireoverride{n}\lstick{$2$}& \ctrl{-1}  & 
    \end{quantikz}
    }
\end{center}

The solution is to entangle the qubits of the ancilla system in a state $\ket{\Phi^+}=\frac{\ket{00}+\ket{11}}{\sqrt{2}}$ so that both measurement results can't be interpreted alone, but xoring the two outputs gives the measurement result: 
\begin{center}
    \scalebox{0.8}{
    \begin{quantikz}
        \lstick[2]{\sc first subsystem}&\wireoverride{n}\lstick{$1$}  &\ctrl{1} & \\
        &\wireoverride{n}\makeebit{}& \targ{}   &\meter{}&\targ{}\setwiretype{c}&\\
        \lstick[2]{\sc second subsystem}&\wireoverride{n}&\targ{}        & \meter{}&\cwbend{-1}\setwiretype{c}\\
        &\wireoverride{n}\lstick{$2$}& \ctrl{-1} & 
    \end{quantikz}
    }
\end{center}
Note that the bit of randomness from the Bell pair is used to mask both measurement outcomes, but simplifies out when xoring them.
The measurement teleportation circuit could also have be obtained by using gate teleportation for one CNOT and simplification of the circuit.

This scheme can be extended to any measurement of Hermitian unitary operators.
\subsection{Measurement teleportation: General introduction}

Consider a bipartite system $A \otimes B$ on which you want to perform a projective measurement of a Hermitian unitary operator $O_A\otimes O_B$ (where $O_A$ and $O_B$ are both hermitian unitary operators) with minimal interaction.
$A$ and $B$ are respectively composed of $n_A$ and $n_B$ qubits.
Measurement teleportation provides a way to do such a measurement using a shared Bell pair.
We note $1$ and $2$ the qubits of this Bell pair such that $A'=A \otimes Q_1$ and $B'=B \otimes Q_2$, with $Q_1$ and $Q_2$ the Hilbert space of the qubits 1 and 2, are two separated systems requiring no quantum interaction between each other (once the Bell pair is formed).
\smallskip

The general scheme is the following:
\begin{center}
    \scalebox{0.8}{
    \begin{quantikz}
        \lstick[2]{\sc first subsystem}&\wireoverride{n}\lstick{$A$}  &\gate{O_A}& & \\
        &\wireoverride{n}\makeebit{}& \ctrl{-1}   &\gate{H}&\meter{}&\targ{}\setwiretype{c}&\\
        \lstick[2]{\sc second subsystem}&\wireoverride{n}&\ctrl{1}       & \gate{H}&\meter{}&\cwbend{-1}\setwiretype{c}\\
        &\wireoverride{n}\lstick{$B$}& \gate{O_B} & &
    \end{quantikz}
    }
\end{center}

Where the line of $A$ and $B$ should be understood as the set of lines with the qubits of each system.

\begin{proof}
    We note $\epsilon_A, \epsilon_B$ the eigenvalues of $O_A$ and $O_B$; the eigenvalue of $O_A \otimes O_B$ are thus $\epsilon_A\epsilon_B$.
    As $O_A$ and $O_B$ are hermitian unitary operators their eigenvalues are either $+1$ or $-1$ so are the eigenvalues of $O_A\otimes O_B$.
    On each system $A$ and $B$ we can build a basis of eigenstates of $O_A$ and $O_B$ that we note $\ket{\phi_{j,\epsilon_{j,i}}^i}$ with $j\in\{A,B\}$, $i\in [\![1,2^{n_j}]\!]$ (we recall that $j$ is made of $n_j$ qubits) and $\epsilon_{j,i}\in\{-1,+1\}$ the eigenvalue of the associated basis state with respect to the operator $O_j$.

    Given a state $\ket{\psi_{AB}}$ of the system $A \otimes B$, we can decompose this state in this basis as:
    \begin{align}\label{AB_full_decomp}
    \begin{split}
    \ket{\psi_{AB}}&=a_{+1,+1}\sum_{\substack{i,j \text{ s.t }\\\epsilon_{A,i}=+1\\\epsilon_{B,j}=+1}}c_{i,j}^{+1,+1}\ket{\phi_{A,\epsilon_{A,i}}^i}\ket{\phi_{B,\epsilon_{B,j}}^j}\\
    &+a_{-1,-1}\sum_{\substack{i,j \text{ s.t }\\\epsilon_{A,i}=-1\\\epsilon_{B,j}=-1}}c_{i,j}^{-1,-1}\ket{\phi_{A,\epsilon_{A,i}}^i}\ket{\phi_{B,\epsilon_{B,j}}^j}\\
    &+a_{-1,+1}\sum_{\substack{i,j \text{ s.t }\\\epsilon_{A,i}=-1\\\epsilon_{B,j}=+1}}c_{i,j}^{-1,+1}\ket{\phi_{A,\epsilon_{A,i}}^i}\ket{\phi_{B,\epsilon_{B,j}}^j}\\
    &+a_{+1,-1}\sum_{\substack{i,j \text{ s.t }\\\epsilon_{A,i}=+1\\\epsilon_{B,j}=-1}}c_{i,j}^{+1,-1}\ket{\phi_{A,\epsilon_{A,i}}^i}\ket{\phi_{B,\epsilon_{B,j}}^j}
    \end{split}
    \end{align}
    For simplicity for any $(k,l)\in\{-1,+1\}^2$ we will note \[\ket{\phi_{k,l}}=\sum_{\substack{i,j \text{ s.t }\\\epsilon_{A,i}=k\\\epsilon_{B,j}=l}}c_{i,j}^{k,l}\ket{\phi_{A,\epsilon_{A,i}}^i}\ket{\phi_{B,\epsilon_{B,j}}^j}\]
    
    The auxiliary system $Q_1 \otimes Q_2$ is in the Bell pair state $\ket{\psi_{12}}=\frac{\ket{00}+\ket{11}}{\sqrt{2}}$.
    At the begining of the circuit, the state of the full system $A'\otimes B'$ is:
    \begin{align*}
        \ket{\psi}=&\ket{\psi_{12}}\ket{\psi_{AB}}\\
        =&\frac{1}{\sqrt{2}}\left(\ket{00\psi_{AB}}+\ket{11\psi_{AB}}\right)\\
        =&\frac{a_{+1,+1}}{\sqrt{2}}\left(\ket{00\phi_{+1,+1}}+\ket{11\phi_{+1,+1}}\right)\\
        &+\frac{a_{-1,-1}}{\sqrt{2}}\left(\ket{00\phi_{-1,-1}}+\ket{11\phi_{-1,-1}}\right)\\
        &+\frac{a_{-1,+1}}{\sqrt{2}}\left(\ket{00\phi_{-1,+1}}
        +\ket{11\phi_{-1,+1}}\right)\\
        &+\frac{a_{+1,-1}}{\sqrt{2}}\left(\ket{00\phi_{+1,-1}}+\ket{11\phi_{+1,-1}}\right)
    \end{align*}

    After the gate (or sequence of gates) C$O_A$ and C$O_B$ we have:
    \begin{align*}
        \ket{\psi'}=&\text{C}O_A\text{C}O_B\ket{\psi_{12}}\ket{\psi_{AB}}\\
        =&\frac{a_{+1,+1}}{\sqrt{2}}\left(\ket{00\phi_{+1,+1}}+\ket{11\phi_{+1,+1}}\right)\\
        &+\frac{a_{-1,-1}}{\sqrt{2}}\left(\ket{00\phi_{-1,-1}}+\ket{11\phi_{-1,-1}}\right)\\
        &+\frac{a_{-1,+1}}{\sqrt{2}}\left(\ket{00\phi_{-1,+1}}-\ket{11\phi_{-1,+1}}\right)\\
        &+\frac{a_{+1,-1}}{\sqrt{2}}\left(\ket{00\phi_{+1,-1}}-\ket{11\phi_{+1,-1}}\right)\\
        =&\frac{a_{+1,+1}}{\sqrt{2}}\left(\ket{00}+\ket{11}\right)\ket{\phi_{+1,+1}}\\
        &+\frac{a_{-1,-1}}{\sqrt{2}}\left(\ket{00}+\ket{11}\right)\ket{\phi_{-1,-1}}\\
        &+\frac{a_{-1,+1}}{\sqrt{2}}\left(\ket{00}-\ket{11}\right)\ket{\phi_{-1,+1}}\\
        &+\frac{a_{+1,-1}}{\sqrt{2}}\left(\ket{00}-\ket{11}\right)\ket{\phi_{+1,-1}}
    \end{align*}
    After the Hadamard gates are applied on the Bell pair qubits the state is:
    \begin{align*}
        \ket{\psi''}=&\frac{\ket{++}+\ket{--}}{\sqrt{2}}\left(a_{+1,+1}\ket{\phi_{+1,+1}}+a_{-1,-1}\ket{\phi_{-1,-1}}\right)\\
        &+\frac{\ket{++}-\ket{--}}{\sqrt{2}}\left(a_{-1,+1}\ket{\phi_{-1,+1}}+a_{+1,-1}\ket{\phi_{+1,-1}}\right)\\
        =&\frac{\ket{00}+\ket{11}}{\sqrt{2}}\left(a_{+1,+1}\ket{\phi_{+1,+1}}+a_{-1,-1}\ket{\phi_{-1,-1}}\right)\\
        &+\frac{\ket{01}+\ket{10}}{\sqrt{2}}\left(a_{-1,+1}\ket{\phi_{-1,+1}}+a_{+1,-1}\ket{\phi_{+1,-1}}\right)
    \end{align*}
    By measuring individually in the $Z$ basis, one cannot deduce any information about the joint state from each individual measurement alone; however, the sum (modulo 2, i.e.\@ the XOR) of these measurements allows us to distinguish the states $\frac{\ket{00}+\ket{11}}{\sqrt{2}}$ and $\frac{\ket{01}+\ket{10}}{\sqrt{2}}$.

    If the XOR equals 0 ($1\oplus1$ or $0\oplus0$), the system $A\otimes B$ is projected onto $\ket{\psi_f} = a_{+1,+1} \ket{\phi_{+1,+1}} + a_{-1,-1} \ket{\phi_{-1,-1}}$.
    By using the decomposition of \autoref{AB_full_decomp} and applying the projector $\frac{\mathbb{I} + O_AO_B}{\sqrt{2}}$ on it, one can easily check that indeed $\ket{\psi_f} = \frac{\mathbb{I} + O_AO_B}{\sqrt{2}} \ket{\psi_{AB}}$.
    In the same way, one can easily check that if the XOR equals 1 then $\ket{\psi_{AB}}$ is projected onto $\ket{\psi_f} = \frac{\mathbb{I} - O_AO_B}{\sqrt{2}} \ket{\psi_{AB}}$.
\end{proof}

Note that the scheme would work with any two qubits maximally entangled state, namely $\mathbb{I}\otimes U \ket{\Phi^+}$ with $U$ a local unitary on $Q_2$, shared between the two subsystem by adding a correction $U^\dagger$ before the controlled gate on qubit $2$.

\subsection{Application to syndrome extraction}
Measurement teleportation can then be applied to measure Pauli operators involving arbitrary number of qubits.
For instance, to perform a 4-qubit $ZZZZ$ measurement, we can use the following: 
\begin{equation}\label{fig:teleported_measurement}
    \scalebox{0.8}{
    \begin{quantikz}
        \lstick[3]{\sc first subsystem}&\wireoverride{n}\lstick{$1$}  &\ctrl{2} & & & \\
        &\wireoverride{n}\lstick{$2$}&           &\ctrl{1}  &         & \\
        &\wireoverride{n}\makeebit{} & \targ{} &\targ{}  &\meter{}&\targ{}\setwiretype{c}&\\
        \lstick[3]{\sc second subsystem}&\wireoverride{n}&\targ{}    &\targ{}  & \meter{}&\cwbend{-1}\setwiretype{c}\\
        &\wireoverride{n}\lstick{$3$}&           & \ctrl{-1} &&  \\
        &\wireoverride{n}\lstick{$4$}&\ctrl{-2}    &         & & 
    \end{quantikz}
    }
\end{equation}

Those schemes can be particularly useful in the context of error-correction, for the stabilizer measurement.
The subtlety in syndrome extraction circuit based on teleported measurement is about error propagation.
Indeed, the fact that we use a Bell pair as an ancilla system modifies the way error propagates during syndrome extraction.
This behavior is code dependent because it is related to logical operators representatives.
In \autoref{fig:seam} we give a way to use measurement teleportation for surface code syndrome extraction while maintaining the full code distance (while using the scheme above would reduce the circuit-level distance).

%% file: appendix_seam.tex
\section{Single column splitting}\label{naive seam}

In this section we show our first attempt to split a rotated surface code with teleported measurement on a single column of ancillae qubit.
In Fig.4 of~\cite{distributed_qec_cosmic_ray}, the authors propose to measure the stabilizer in an analogous way by using teleported cnot to prepare the Bell pairs ancillae.
However, this straightforward adaptation of lattice surgery doesn't preserve the code distance, requiring thus a more careful design as the one we propose in \autoref{fig:seam}.
The stabilizers and the teleported syndrome ancillae are shown in \autoref{fig:naive seam}.
The depolarizing noise model used in the circuit-level simulations is described in \autoref{annexe:noise model}.

\begin{figure}[h]
    \centering
    \includegraphics[width=0.5\linewidth]{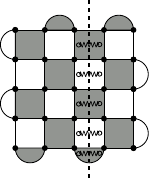}
    \caption{Rotated surface code with only one column of teleported syndrome measurement. Dashed line represent the separation between the two processors. As previously used, empty dots connected by a sinusoidal line correspond to ancillary qubits prepared in a Bell state. Ancillary qubits used for syndrome extraction on a single processor are not represented.}
    \label{fig:naive seam}
\end{figure}

It is clear from \autoref{fig:naive seam} that each qubit of the Bell pair has to do a CNOT with the two data qubits that are on their side of the split between processors during syndrome extraction, following the scheme presented in \autoref{fig:teleported_measurement}.
Therefore, hook errors will always affect two data qubits vertically aligned.
Only $\left\lceil \frac{1}{2}\left\lfloor \frac{d+1}{2}\right\rfloor \right\rceil$ $X$ errors during syndrome extraction are required to have more than $\left\lfloor \frac{d+1}{2}\right\rfloor$ $X$ errors on data qubits leading to a logical error.
This noise propagation scheme effectively affects the distance of the code that is now: 
\[d_{\text{eff}}=2\left\lceil \frac{1}{2}\left\lfloor \frac{d+1}{2}\right\rfloor \right\rceil -1\]

This hook error process can either come from Bell pair preparation $X$ errors (see \autoref{Annexe noise Bell}) or from a $X\otimes X$ error in the two qubits depolarizing channel of any of the first CNOT on each side of the Bell pair during syndrome extraction.

Simulation with only Bell state noisy preparation and perfect Clifford operation clearly exhibits this halved effective distance, see \autoref{fig:naive seam fit}.

\begin{figure}[h]
    \centering
    \includegraphics[width=0.9\linewidth]{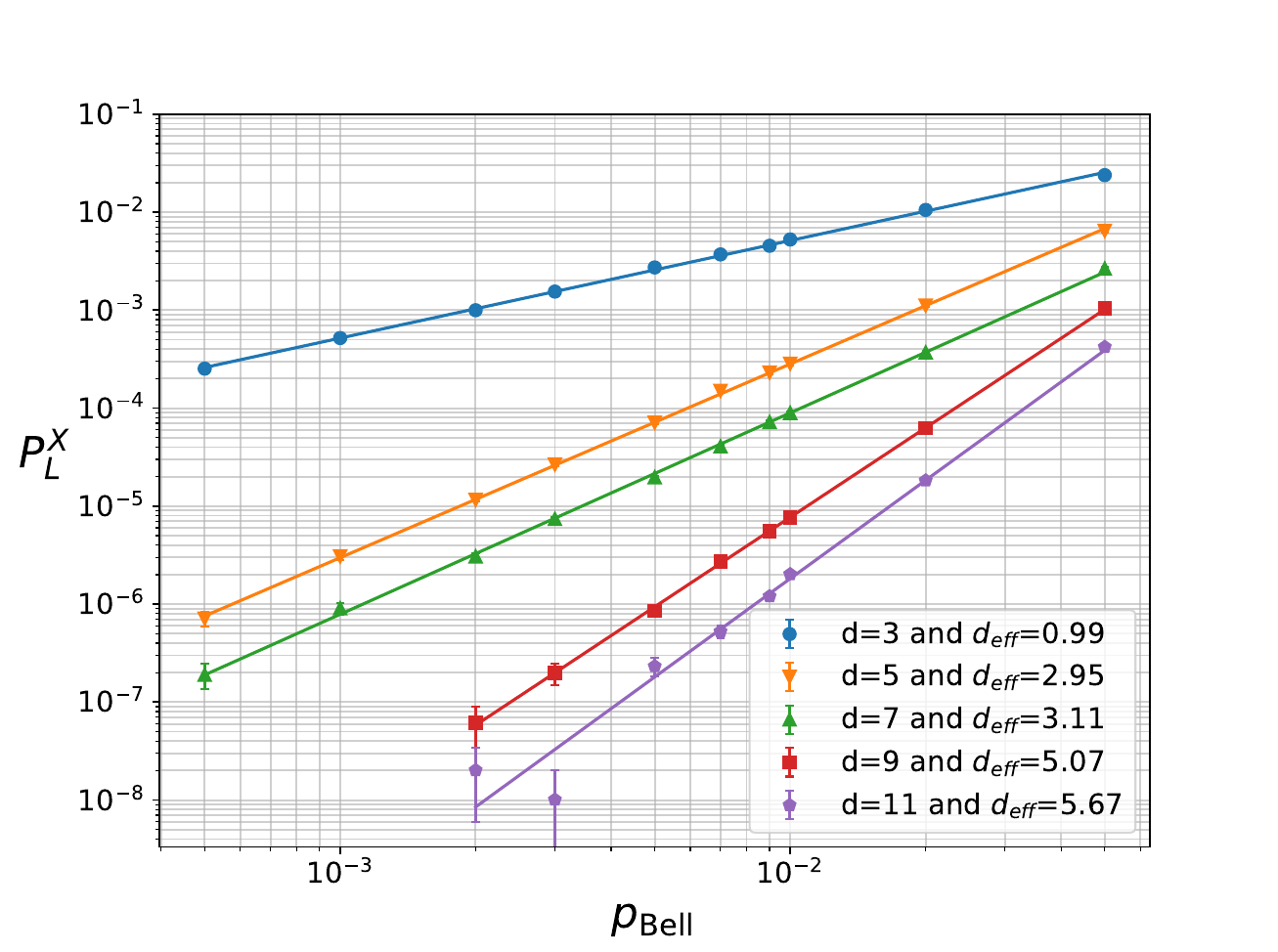}
    \caption{Circuit-level memory experiment simulation of \autoref{fig:naive seam} with only Bell state preparation noise.
    Effective distance $d_\text{eff}$ is obtained by a linear regression over the logarithm of the logical error rate, whose slope is $\frac{d_\text{eff}+1}{2}$.
    It shows a halved effective distance compared to the distance of the rotated surface code.}
    \label{fig:naive seam fit}
\end{figure}

%% file: appendix_ansatz.tex
\section{Analytical bound on logical error rate}\label{annexe: ansatz}


The aim of this section is to derive an ansatz for the logical error rate in the distributed scenario. Specifically, we are interested in the effect of the physical error rate $p_{\text{Bell}}$ on the $X$ logical error rate (denoted $P^X_L$) because we consider here a seam\footnote{Note that $XX$ and $ZZ$ stabilize a Bell pair, which is why errors during the Bell pair preparation can always be considered as having affected the qubit of the pair that will interact with one data qubit (the one on the seam), and not the one interacting with 3 data qubits.} (see \autoref{fig:seam} for the definition of seam qubits) aligned with $X_L$, the logical $X$ operator.
$P^Z_L$ is only very slightly affected as any --minimal length-- representative of the $Z_L$ operator (i.e.\@ any horizontal line of data qubits in our stabilizer convention) will contain only one data qubit of the seam. 
Therefore, the key performance assessment criterion is a memory experiment of a logical state $\ket{0}_{L}$.

To achieve our aim, we adapt the bound derived in~\cite{ramette2023faulttolerantconnectionerrorcorrectedqubits} to the rotated surface code with measurement teleportation.
Concretely, changing the code changes the lattice over which we will count the possible logical error mechanisms.
As highlighted in~\cite{orourke2024comparepairrotatedvs}, rotated and unrotated surface codes have slightly different phenomenological matching graphs.
This feature will appear later in the computation and slightly modify the expression of the bound presented in~\cite{ramette2023faulttolerantconnectionerrorcorrectedqubits}.


\subsection{Simplified noise model and path counting analogy}\label{sec:Analytical bound}


We work in the simplified situation of a phenomenological noise model of parameter $p_{\text{B}}$ (B for ``bulk'') for both noise on data qubits and syndrome measurement.
This parameter aims at capturing the effect of the physical error rate $p$ from the circuit-level simulations.
The effect of Bell pair infidelity is modeled by an error rate $p_{\text{S}}$ (S for ``seam'') for syndromes that are measured by a Bell pair and for data qubits on the seam.
Note that in our circuit-level simulation, $p_\text{Bell}$ is applied only to the Bell pair, but here we study a phenomenological model where $p_\text{S}$ captures the effect of the propagation of $p_\text{Bell}$ through hook errors.

As a decoder, we consider a Minimum Weight Perfect Matching (MWPM) algorithm~\cite{higgott2021pymatchingpythonpackagedecoding} that uses the error matching graph. 
\begin{figure}[h!]
    \centering
    \includegraphics[width=0.4\textwidth]{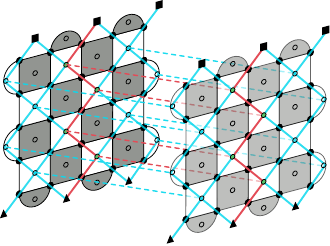}
    \caption{Distance 5 rotated surface code $X$ errors matching graph for phenomenological noise model.
    Surface code patches are only displayed for facilitating reader's understanding, the matching graph edges are only the colored ones.
    Dashed-edges are time-like errors (measurement errors). Red edges have error rates $p_{\text{S}}$ whereas blue edges have error rates $p_{\text{B}}$. Black dots are data qubits, while empty and green ones correspond to the stabilizer measurements (green ones are $Z$ checks involving a Bell pair). The triangles and squares represent the two $X$ boundaries.
    Note that red edges qubits are the one on the central column of the zig-zag QPU separation on \autoref{fig:seam}.
    }\label{fig:seam in 3d}
\end{figure}
\autoref{fig:seam in 3d} shows such a graph for the $X$ errors under a phenomenological noise.
The vertices of the matching graph are $Z$ checks.
A check is triggered if the associated stabilizer eigenvalue is different from its value at the previous round.
Specifically, an error on a data qubit triggers the checks associated to the stabilizers sharing this data qubit, that is checks connected by a continuous edge (blue or red depending on the noise probability $p_{\text{S}}$ or $p_{\text{B}}$).
An error on a stabilizer measurement triggers the same check at two successive rounds, that is checks connected by a dashed edge.
The two $X$ boundaries are represented on \autoref{fig:seam in 3d} by respectively squares and triangles.
The set of blue (resp.\@ red) edges is called $\{\text{Bulk}\}$ (resp.\@$\{\text{Seam}\}$).

For a given memory experiment, we define $\{E\}$ as the set of edges where $X$ errors happened, and $\{R\}$ as the set of edges flagged by the decoder as having caused the syndrome pattern; $\{R\}$ is latter referred to as the recovery operation.
Our protected state undergoes a logical $X$ error if the sum (symmetric difference) $\{E \oplus R\}=(\{E\}\backslash\{R\})\cup (\{R\}\backslash\{E\})$ of both operations is a non-trivial connected set of edges from one $X$ boundary to the other one.

For a set of edges $\{S\}$, we denote $\partial (\{S\})$ the boundaries of this set, that is 
the set of check vertices connected to an odd number of 
edges from $\{S\}$.
Note that omitting the brackets means the cardinality of: $S=\Card{\{S\}}$.

\smallskip
The construction of the bound is obtained from the following structure:
\begin{enumerate}
    \item Given a connected set of edges $\{\gamma\}$ on the matching graph represented in \autoref{fig:seam in 3d}, we compute the probability that a set of errors $\{E\}$ such that $\{E \oplus R\}$ contains $\{\gamma\}$ happens.
    We write this probability $\mathbb{P}\left( \{E \oplus R\}\supseteq \{\gamma\} \right)$.
    We call such $\{\gamma\}$ a ``path''. 
    \item We parameterize and bound the number of paths that give a logical error.
    \item The final step consists in using these parameters to sum the probability computed in 1.\@ over all the paths leading to a logical error.
\end{enumerate}

\subsection{Probability of failure with a given path}
The derivation of this subsection has been originally done in~\cite{ramette2023faulttolerantconnectionerrorcorrectedqubits}.
We report it here for completeness.
Given a set of edges $\{E\}$ its associated probability is $\mathbb{P
}(\{E\})= \prod_{i}(1-p_i)\prod_{i\in \{E\}}\frac{p_i}{1-p_i}\propto\prod_{i\in \{E\}}\frac{p_i}{1-p_i}$
where the first product is over all edges of the graph, and $p_i$ is the error probability associated to the edge indexed by $i$.
From there one can define the weight of an edge $i$ as being $\text{wt}(\{i\})=\log\left(\frac{
1-pi}{p_i}\right)$.
The weight of a set of edges is just the sum of the edges weights.

MWPM decoder returns the recovery operation that minimizes the recovery weight $\text{wt}(\{R\})$ (i.e.\@ find the most probable recovery set), while ensuring $\partial(\{R\}) = \partial(\{E\})$.

We are interested in paths $\{\gamma\}$ that are closed, i.e.\@, paths such that $\partial(\{\gamma\})=\emptyset$ and non-trivial, i.e.\@, paths that link the two $X$ boundaries.
We note $\Lambda$ this set of paths.

We consider a path $\{\gamma\}\in \Lambda$ such that $\{\gamma\}\subset\{E \oplus R\}$.
Hence $\text{wt}(\{\gamma\})= \text{wt}(\{E\}\cap\{\gamma\})+\text{wt}(\{R\}\cap\{\gamma\})$.

On the other hand, $\text{wt}(\{R\}\cap\{\gamma\})<\text{wt}(\{E\}\cap\{\gamma\})$ because otherwise one could build a smaller weight recovery operation $\{R'\}=\left(\{R\}\backslash\{\gamma\}\right)\bigcup\left(\{E\}\cap\{\gamma\}\right)$ with $\partial(\{R'\}) = \partial(\{E\})$ because $\partial(\{E\}\cap\{\gamma\})=\partial(\{R\}\cap\{\gamma\})$ as $\{\gamma\}$ is a closed path.

Therefore, we have 
$\text{wt}(\{\gamma\})\leq 2\text{wt}(\{E\}\cap\{\gamma\})$ from which we directly deduce the following inequality (by simply rewriting the definition of the weights):
\begin{multline}\label{ineq:MWPM}
{\left(\frac{p_{\text{S}}}{1-p_{\text{S}}}\right)}^{\gamma_{\text{SE}}}
{\left(\frac{p_{\text{B}}}{1-p_{\text{B}}}\right)}^{\gamma_{\text{BE}}} \\
\leq
{\left(\frac{p_{\text{S}}}{1-p_{\text{S}}}\right)}^{\frac{\gamma_{\text{S}}}{2}}
{\left(\frac{p_{\text{B}}}{1-p_{\text{B}}}\right)}^{\frac{\gamma_{\text{B}}}{2}}
\end{multline}
where we used the notations $\{\gamma_\text{S}\}=\{\gamma\}\cap\{\text{Seam}\}$, $\{\gamma_\text{B}\}=\{\gamma\}\cap\{\text{Bulk}\}$, $\{\gamma_\text{SE}\}=\{E\}\cap\{\gamma\}\cap\{\text{Seam}\}$, $\{\gamma_\text{BE}\}=\{E\}\cap\{\gamma\}\cap\{\text{Bulk}\}$. Here $\{\text{Bulk}\}$ and $\{\text{Seam}\}$ refer to the bulk qubit and seam qubit ensembles.

\bigskip

Independently, given a path $\{\gamma\}\in\Lambda$, the probability that a set of errors $\{E\}$ generated on the matching graph has exactly $\gamma_\text{BE}, \gamma_\text{SE} $ overlaps with $\{\gamma_\text{B}\}$, $\{\gamma_\text{S}\}$ is:
\begin{multline*}
    \mathbb{P}(\gamma_\text{BE}, \gamma_\text{SE})=\binom{\gamma_\text{S}}{\gamma_\text{SE}}\binom{\gamma_\text{B}}{\gamma_\text{BE}} p_{\text{S}}^{\gamma_\text{SE}}{(1-p_{\text{S}})}^{\gamma_\text{S}-\gamma_\text{SE}}\\ p_{\text{B}}^{\gamma_\text{BE}}{(1-p_{\text{B}})}^{\gamma_\text{B}-\gamma_\text{BE}}
\end{multline*}
If one imposes that $\{\gamma\}\subset\{E \oplus R\}$, \autoref{ineq:MWPM} applies and we can bound  $\mathbb{P}\left( \{E \oplus R\}\supseteq \{\gamma\} \right)$ (for a given $\{\gamma\}$) using the expression of $\mathbb{P}(\gamma_\text{BE}, \gamma_\text{SE})$ by summing over all possible $\gamma_\text{BE}, \gamma_\text{SE}$:
\begin{multline*}
\mathbb{P}\left( \{E \oplus R\}\supseteq \{\gamma\} \right) = \sum_{\substack{\gamma_\text{SE},\gamma_\text{BE}\\\text{s.t.\@ }\{\gamma\}\subset\{E \oplus R\}}}
\mathbb{P}(\gamma_\text{BE}, \gamma_\text{SE})\\
\leq \sum_{\substack{\gamma_\text{SE},\gamma_\text{BE}\\\text{s.t.\@ }\{\gamma\}\subset\{E \oplus R\}}}
\binom{\gamma_\text{S}}{\gamma_\text{SE}}\binom{\gamma_\text{B}}{\gamma_\text{BE}}
p_{\text{S}}^{\gamma_\text{S}/2} p_{\text{B}}^{\gamma_\text{B}/2}\\
\leq \sum_{\substack{1\leq\gamma_\text{SE}\leq \gamma_{\text{S}}\\1\leq\gamma_\text{BE}\leq\gamma_{\text{B}}}}
\binom{\gamma_\text{S}}{\gamma_\text{SE}}\binom{\gamma_\text{B}}{\gamma_\text{BE}}
p_{\text{S}}^{\gamma_\text{S}/2} p_{\text{B}}^{\gamma_\text{B}/2}\\
\leq2^{\gamma_\text{S}+\gamma_\text{B}}p_{\text{S}}^{\gamma_\text{S}/2}p_{\text{B}}^{\gamma_\text{B}/2}.
\end{multline*}

The probability that a logical error happens is the sum of this probability over all possible non-trivial closed paths going from the bottom of the lattice to the top
\begin{multline}\label{pfail as sum over non trivial path}
P_{\text{fail}} = \sum_{\{\gamma\} \in \Lambda}\mathbb{P}\left( \{E \oplus R\}\supseteq \{\gamma\} \right) \\
\leq \sum_{\{\gamma\}\in \Lambda} 2^{\gamma_\text{S}+\gamma_\text{B}}p_{\text{S}}^{\gamma_\text{S}/2}p_{\text{B}}^{\gamma_\text{B}/2}.
\end{multline}

\subsection{Counting non-trivial closed paths}
We now count a larger type of path: the self-avoiding walks.
A self-avoiding walk (SAW) is a sequence of moves on a lattice that starts at a given point and proceeds to new points, with the constraint that the walk cannot visit the same point more than once.
\begin{lemma}
    From any non-trivial closed path $\{\gamma\}$ one can always extract a self-avoiding-walk $\{\gamma_{\text{SAW}}\}$ joining the two opposite logical boundaries.
    We will note the set of SAW joining the two opposite logical boundaries as $\Lambda_{\text{SAW}}$.
\end{lemma}
\begin{proof}
    Given a non-trivial closed path $\{\gamma\}$, $\{\gamma\}$ is a connected set so it is always possible to extract a path that joins the two logical boundaries from it.
    Let's call $\{\gamma_{\text{SAW}}\}$ the minimal length path linking the two logical boundaries that one can extract from $\{\gamma\}$.
    Then $\{\gamma_{\text{SAW}}\}$ is necessarily a self-avoiding-walk.
    Indeed the opposite would mean that at some point $\{\gamma_{\text{SAW}}\}$ visits the same point two times which means that $\{\gamma_{\text{SAW}}\}$ has a trivial close loop that we can call $\{\gamma_{\text{loop}}\}$.
    If that is the case, $\{\gamma_{\text{SAW}}\}\backslash\{\gamma_{\text{loop}}\}$ is a shorter path, that can be extracted from $\{\gamma\}$, linking the two logical boundaries which contradicts the definition of $\{\gamma_{\text{SAW}}\}$.
\end{proof}
We have the following event inclusion: 
For any $\{\gamma\}\in\Lambda$ and $\{\gamma_{\text{SAW}}\}\in\Lambda_{\text{SAW}}$ extracted from it,
\[\left( \{E \oplus R\}\supseteq \{\gamma\} \right) \subset\left( \{E \oplus R\}\supseteq \{\gamma_{\text{SAW}}\} \right). 
 \]
Therefore, in \autoref{pfail as sum over non trivial path} it is sufficient to sum over the $\{\gamma_{\text{SAW}}\}$ to have a bound on the logical error probability: 

\begin{multline}\label{pfail as sum over SAW}
P_{\text{fail}} = \sum_{\{\gamma\}\in \Lambda}\mathbb{P}\left( \{E \oplus R\}\supseteq \{\gamma\} \right) \\
\leq\sum_{\{\gamma\}\in \Lambda_{\text{SAW}}}\mathbb{P}\left( \{E \oplus R\}\supseteq \{\gamma\} \right) \\
\leq \sum_{\{\gamma\}\in \Lambda_{\text{SAW}}} 2^{\gamma_\text{S}+\gamma_\text{B}}p_{\text{S}}^{\gamma_\text{S}/2}p_{\text{B}}^{\gamma_\text{B}/2}.
\end{multline}
On the other hand,
\begin{equation*}
\Lambda_{\text{SAW}}\subseteq\bigcup_{l\geq d}\text{SAW}(l)
\end{equation*}
where $\text{SAW}(l)$ is the set of self-avoiding walks of length $l$.
For a self-avoiding walk of length $l\geq d$ in a $45°$ tilted square lattice of dimension $D$, $\abs{\text{SAW}(l)} \leq 2D{(2D-1)}^{l-1}$ because at each step, the walk has at most $2D-1$ possibilities to continue.
This bound disregards lattice border effects and the inability to reuse an edge taken by earlier steps, except for the most recent one.

In our situation, as the syndrome measurements are noisy, we are repeating the error correcting cycle $\mathcal{O}(d)$ times with $d$ being the distance.
Therefore, as explicit on \autoref{fig:seam in 3d}, the set of errors in the bulk and in the seam are respectively living in spaces of dimension $D_\text{B}=(2+1)=3$ and $D_\text{S}=(1+1)=2$.

Typically, as $p_{\text{S}}>p_{\text{B}}$, the path will be likely to stay on the seam.
However, in the regime where $p_{\text{B}}$ is not negligible compared to $p_{\text{S}}$, we have to take into account situations where the path leaves the seam for a number of steps or even until reaching the boundary.

\begin{figure}
    \centering
    \includegraphics[width=0.5\linewidth]{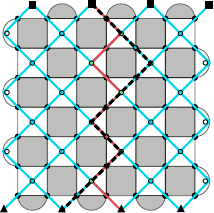}
    \caption{Matching graph slice at fixed time with a path $\{\gamma\}\in\Lambda $ represented as a black dashed line. For this path $\gamma_\text{B}=3$, $\gamma_\text{S}=4$, $C=1$, $l_0=0$, $l_1=2$, $l_2=1$ (definitions in the text).
    Note that here $\{\gamma\}$ is space-like, in general it can also take time-like edges represented as dashed edges on \autoref{fig:seam in 3d}.}\label{fig:matching graph 2d}
\end{figure}

We will parameterize the set of self-avoiding walks so that we can keep track of $\gamma_\text{B}$ and $\gamma_\text{S}$.
We will call $C$ the number of full excursions (hopping off the seam and then hopping back on the seam later).
Excursions that end on a border of the patch (squares and triangles in \autoref{fig:matching graph 2d}) are not counted in $C$.
Given a number of seam edges $\gamma_\text{S}$, there are $\binom{\gamma_\text{S}-1}{C}$ ways to choose the locations of $C$ excursions (a location being here the index of the seam edge after which the excursion starts, and not a coordinate on the matching graph).
Each excursion's bulk segment length is denoted $l_k$, $k$ being the index of the excursion.
$l_0$ and $l_{C+1}$ are the -- eventually null -- length of the bulk segments before joining the seam and after leaving it for the last time.
Then we have $\gamma_\text{B}=\sum_{k=0}^{C+1}l_k$.

For each of the $\gamma_\text{S}$ steps on the seam, if the previous step was on the seam we have $\mu_\text{S}=2D_\text{S}-1$ possibilities; otherwise we have $2D_\text{S}$ options.
For each of the steps in the $\{l_k\}$ we have less than $\mu_\text{B}=2D_\text{B}-1$ as well except for the first edge when leaving the seam that has $2(D_\text{B}-D_\text{S})$ ways to do so.
Hence, the number of self-avoiding walks starting from a given point, with $C$ excursions of lengths $\{l_k\}$ and $\gamma_\text{S}$ edges on the seam is bounded as
\begin{equation*}
n_{\text{SAW}}\left(C,\{l_k\},\gamma_\text{S}\right) \leq \binom{\gamma_\text{S}-1}{C} a^C \mu_\text{B}^{\sum_k l_k}\mu_\text{S}^{\gamma_\text{S}}
\end{equation*}
with $a=\frac{4(D_\text{B}-D_\text{S})D_\text{S}}{\mu_\text{S} \mu_\text{B}}=\frac{1}{2}$.

The probability of generating $\{E \oplus R\}$ with $\{\gamma\}$ included in it given the fact that $\{\gamma\}$ has parameters $C, \{l_k\}$, $\gamma_{\text{S}}$ is
\begin{multline*}
\mathbb{P}_{\text{SAW}}(C,\gamma_\text{S},\{l_k\}) \\
\leq N n_{\text{SAW}}\left(C,\{l_k\},\gamma_\text{S}\right) 2^{\gamma_\text{S}+\gamma_\text{B}}p_{\text{S}}^{\gamma_\text{S}/2}p_{\text{B}}^{\gamma_\text{B}/2} 
\end{multline*}
with $N$ being the number of possible starting points for the SAW.
Anyway, we can say that $N\leq \text{poly}(d)$ as there are only a polynomial number of points on the spacetime lattice.

Now we must sum over all possible parametrizations:
\begin{equation}\label{difference}
P_{\text{fail}} \leq \sum_{\gamma_\text{S}\geq 0}\sum_{C=0}^{\gamma_\text{S}-1}\sum_{\substack{\{l_k\geq 0\} \\ \gamma_\text{S}+\sum_{k=0}^{C+1}l_k\geq d}}\mathbb{P}_{\text{SAW}}(C,\gamma_\text{S},\{l_k\})
\end{equation}
where the notation $\{l_k\geq 0\}$ indicates that the sum is over all the possible excursion paths $\{l_k\}$ with non-negative cardinal $l_k$ for each $k$.

Using the previous bound when separating the cases $\gamma_\text{S}=0$ and $\gamma_\text{B}=0$ (but keeping the case $C=0, \gamma_\text{S}>0 \text{ with } l_0>0 \text{ or } l_1>0$ in the sum corresponding to situations with no excursion but starting outside of the seam and joining the seam or leaving the seam and finishing in the bulk) gives us:
\begin{multline}\label{annexe: ansatz, intermediate upper bond}
    P_{\text{fail}}/\text{poly}(d) \leq \sum_{\gamma_\text{S}\geq d} {\left(\frac{p_{\text{S}}}{p_{\text{S}}^*}\right)}^{\gamma_\text{S}/2}
    +
    \sum_{\gamma_\text{B}\geq d} {\left(\frac{p_{\text{B}}}{p_{\text{B}}^*}\right)}^{\gamma_\text{B}/2}\\
    +
    \sum_{\gamma_\text{S}\geq1}{\left(\frac{p_{\text{S}}}{p_{\text{S}}^*}\right)}^{\gamma_\text{S}/2}\sum_{C=0}^{\gamma_\text{S}-1}\binom{\gamma_\text{S}-1}{C}a^C\\
    \sum_{\substack{\{\forall k| 1\leq k \leq C,~l_k\geq2\}\\ \{l_0,~l_{C+1}~\geq 0\}\\ \gamma_\text{S}+\sum_{k=0}^{C+1}l_k\geq d}}{\left(\frac{p_{\text{B}}}
    {p_{\text{B}}^*}\right)}^{\frac{1}{2}\sum_k l_k}
\end{multline}
where $p_{\text{S}}^* = {(2\mu_\text{S})}^{-2}$, $p_{\text{B}}^* = {(2\mu_\text{B})}^{-2}$.


Now we can rewrite the last term as: 
\begin{multline}\label{annexe: ansatz last term developped}
\sum_{\substack{\{\forall k| 1\leq k \leq C,~l_k\geq2\}\\ \{l_0,~l_{C+1}~\geq0\}\\ \gamma_\text{S}+\sum_{k=0}^{C+1}l_k\geq d}}{\left(\frac{p_{\text{B}}}{p_{\text{B}}^*}\right)}^{\frac{1}{2}\sum_k l_k} \\
=\sum_{\substack{\{\forall k| 0\leq k \leq C+1,~l_k'\geq0\}\\  \gamma_\text{S}+2C+\sum_{k=0}^{C+1}l_k'\geq d}} {\left(\frac{p_{\text{B}}}{p_{\text{B}}^*}\right)}^{\frac{1}{2}\sum_k l_k'+C}
\\
= {\left(\frac{p_{\text{B}}}{p_{\text{B}}^*}\right)}^{\frac{d-\gamma_\text{S}-2C}{2}\bigr\rvert_{>0}} \sum_{\{\forall k| 0\leq k \leq C+1,~l_k^*\geq0\}} {\left(\frac{p_{\text{B}}}{p_{\text{B}}^*}\right)}^{\frac{1}{2}\sum_k l_k^*+C}
\\
= {\left(\frac{p_{\text{B}}}{p_{\text{B}}^*}\right)}^{\frac{d-\gamma_\text{S}-2C}{2}\bigr\rvert_{>0}} {\left(\frac{p_{\text{B}}}{p_{\text{B}}^*}\right)}^C \sum_{\{l_k^*\geq0\}} {\left(\frac{p_{\text{B}}}{p_{\text{B}}^*}\right)}^{\frac{1}{2}\sum_k l_k^*}.
\end{multline}
First equality comes from the change of variable $l_k=2+l_k'$ for $k \in \doublebrackets{1, C}$, and $l_k = l'_k$ for $k=0$ and $k=C+1$.
For the second equality, note that the condition $\gamma_\text{S}+2C+\sum_{k=0}^{C+1}l_k'\geq d$ is only restrictive when $d - \gamma_\text{S} - 2C \geq 0$.
Let us introduce a second change of variable: $l_k'=d_k+l_k^*$ for $k \in \doublebrackets{0, C+1}$ with $l_k$ and $d_k$ arbitrary non-negative integers, and $d_k$ such that $\sum_k d_k={d-\gamma_\text{S}-2C}\bigr\rvert_{>0}$ where $x\bigr\rvert_{>0} =
\begin{dcases*}
    x & if $x>0$ \\
    0 & else 
\end{dcases*}$.
Using this change of variable gives the second equality.

For simplicity, we will write $l_k $ instead of $l_k^*$ from now on.
One can consider that the sum over each $l_k$ runs towards infinity; in reality, it is limited by the available space on the lattice, but very high orders will be negligible as $\left(\frac{p_{\text{B}}}{p_{\text{B}}^*}\right)<1$ in the regime of interest.
We can also compute:
\begin{align*}
\sum_{\{l_k\geq0\}} {\left( \frac{p_{\text{B}}}{p_{\text{B}}^*}\right)}^{\frac{1}{2}\sum_k l_k}
&= \sum_{\{l_k\geq0\}} \left[ \prod_{k=0}^{C+1} {\left( \frac{p_{\text{B}}}{p_{\text{B}}^*}\right)}^{\frac{1}{2} l_k} \right] \\
&= \prod_{k=0}^{C+1} \left[ \sum_{l_k\geq0} {\left( \frac{p_{\text{B}}}{p_{\text{B}}^*}\right)}^{\frac{1}{2} l_k} \right] \\
&= \prod_{k = 0}^{C+1} {\left(\frac{1}{1-\sqrt{\frac{p_{\text{B}}}{p_{\text{B}}^*}}}\right)}\\
&= {\left(\frac{1}{1-\sqrt{\frac{p_{\text{B}}}{p_{\text{B}}^*}}}\right)}^{C+2}
\end{align*}
Note that we go from the first to the second equality by distributing the product, as in~\cite[Eq.\,(17)]{ramette2023faulttolerantconnectionerrorcorrectedqubits}).

Now by injecting this relation in \autoref{annexe: ansatz last term developped}, the third term in the upper bound \autoref{annexe: ansatz, intermediate upper bond} of $P_{\text{fail}}/\text{poly}(d)$ is bounded by:
\begin{multline}\label{annexe: ansatz developpement inegalite}
    \sum_{\gamma_\text{S}\geq1} {\left(\frac{p_{\text{S}}}{p_{\text{S}}^*}\right)}^{\gamma_\text{S}/2}
    \sum_{C=0}^{\gamma_\text{S}-1}\binom{\gamma_\text{S}-1}{C}a^C \\
    \sum_{\substack{\{\forall k| 1\leq k \leq C,~l_k\geq2\}\\
    \{l_0,~l_{C+1}~\geq0\}\\
    \gamma_\text{S}+\sum_{k=0}^{C+1}l_k\geq d}} 
    {\left(\frac{p_{\text{B}}}{p_{\text{B}}^*}\right)}^{\frac{1}{2}\sum_k l_k}\\
    \leq
    \sum_{\gamma_\text{S}\geq1} {\left(\frac{p_{\text{S}}}{p_{\text{S}}^*}\right)}^{\gamma_\text{S}/2}
    \sum_{C=0}^{\gamma_\text{S}-1}\binom{\gamma_\text{S}-1}{C}a^C 
    {\left(\frac{p_{\text{B}}}{p_{\text{B}}^*}\right)}^{\frac{d-\gamma_\text{S}-2C}{2}\bigr\rvert_{>0}}  {\left(\frac{p_{\text{B}}}{p_{\text{B}}^*}\right)}^C \\
  {\left(\frac{1}{1-\sqrt{\frac{p_{\text{B}}}{p_{\text{B}}^*}}}\right)}^{C+2}
    \\
    \leq
     \sum_{\gamma_\text{S}\geq 1} {\left(\frac{p_{\text{S}}}{p_{\text{S}}^*}\right)}^{\gamma_\text{S}/2}
     {\left(\frac{p_{\text{B}}}{p_{\text{B}}^*}\right)}^{\frac{d-\gamma_\text{S}}{2}}
     \sum_{C=0}^{\gamma_\text{S}-1}\binom{\gamma_\text{S}-1}{C}a^C 
      \\ {\left(\frac{1}{1-\sqrt{\frac{p_{\text{B}}}{p_{\text{B}}^*}}}\right)}^{C+2}
\end{multline}
where to get rid of the positive part in the exponent of ${\left(\frac{p_{\text{B}}}{p_{\text{B}}^*}\right)}^{\frac{d-\gamma_\text{S}-2C}{2}\bigr\rvert_{>0}}$ we used the fact that:
\begin{multline*}
\forall C>\frac{d-\gamma_\text{S}}{2},\text{ if } p_{\text{B}}<p_{\text{B}}^* \text{ then,}\\  1 = {\left(\frac{p_{\text{B}}}{p_{\text{B}}^*}\right)}^{\frac{d-\gamma_\text{S}-2C}{2}\bigr\rvert_{>0}} < {\left(\frac{p_{\text{B}}}{p_{\text{B}}^*}\right)}^{\frac{d-\gamma_\text{S}-2C}{2}}.
\end{multline*}

Using the binomial theorem one can rewrite the inner sum in the last inequality of \autoref{annexe: ansatz developpement inegalite} as,
\begin{equation*}
\sum_{C=0}^{\gamma_\text{S}-1}\binom{\gamma_\text{S}-1}{C}a^C
{\left(\frac{1}{1-\sqrt{\frac{p_{\text{B}}}{p_{\text{B}}^*}}}\right)}^{C}
= {\left(1+\frac{a}{1-\sqrt{\frac{p_{\text{B}}}{p_{\text{B}}^*}}}\right)}^{\gamma_\text{S}-1}.
\end{equation*}

By writing a ``pseudo-threshold'' for the seam as: 
\begin{equation}\label{new threshold}
    p_{\text{S}}^{**}= p_{\text{S}}^* {\left(1+\frac{a}{1-\sqrt{\frac{p_{\text{B}}}{p_{\text{B}}^*}}}\right)}^{-2}
\end{equation}
 and denoting $\beta(p_{\text{B}})= {\left(\frac{1}{1-\sqrt{\frac{p_{\text{B}}}{p_{\text{B}}^*}}}\right)}^{2} {\left(1+\frac{a}{1-\sqrt{\frac{p_{\text{B}}}{p_{\text{B}}^*}}}\right)}^{-1}$, we can then 
 rewrite \autoref{annexe: ansatz developpement inegalite} as: 
 \begin{multline*}
    \sum_{\gamma_\text{S}\geq1} {\left(\frac{p_{\text{S}}}{p_{\text{S}}^*}\right)}^{\gamma_\text{S}/2}
    \sum_{C=0}^{\gamma_\text{S}-1}\binom{\gamma_\text{S}-1}{C}a^C \\
    \sum_{\substack{\{\forall k| 1\leq k \leq C,~l_k\geq2\}\\
    \{l_0,~l_{C+1}~\geq0\}\\
    \gamma_\text{S}+\sum_{k=0}^{C+1}l_k\geq d}} 
    {\left(\frac{p_{\text{B}}}{p_{\text{B}}^*}\right)}^{\frac{1}{2}\sum_k l_k}\\
    \leq
     \sum_{\gamma_\text{S}\geq 1} {\left(\frac{p_{\text{S}}}{p_{\text{S}}^*}\right)}^{\gamma_\text{S}/2} {\left(\frac{p_{\text{B}}}{p_{\text{B}}^*}\right)}^{\frac{d-\gamma_\text{S}}{2}}
     {\left(1+\frac{a}{1-\sqrt{\frac{p_{\text{B}}}{p_{\text{B}}^*}}}\right)}^{\gamma_\text{S}-1}
     \\ {\left(\frac{1}{1-\sqrt{\frac{p_{\text{B}}}{p_{\text{B}}^*}}}\right)}^{2}\\
     \leq
     \beta(p_{\text{B}}) \sum_{\gamma_\text{S}\geq1} {\left(\frac{p_{\text{S}}}{p_{\text{S}}^{**}}\right)}^{\gamma_\text{S}/2} {\left(\frac{p_{\text{B}}}{p_{\text{B}}^*}\right)}^{\frac{d-\gamma_\text{S}}{2}}.
\end{multline*}

Finally, the full bound \autoref{annexe: ansatz, intermediate upper bond} can be expressed in the following compact form:
\begin{multline}\label{eq:analytical bound}
P_{\text{fail}}/\text{poly}(d)\leq \sum_{\gamma_\text{S}\geq d} {\left(\frac{p_{\text{S}}}{p_{\text{S}}^*}\right)}^{\gamma_\text{S}/2} + \sum_{\gamma_\text{B}\geq d} {\left(\frac{p_{\text{B}}}{p_{\text{B}}^*}\right)}^{\gamma_\text{B}/2} \\
+ \beta(p_{\text{B}}) \sum_{\gamma_\text{S}\geq1} {\left(\frac{p_{\text{S}}}{p_{\text{S}}^{**}}\right)}^{\gamma_\text{S}/2} {\left(\frac{p_{\text{B}}}{p_{\text{B}}^*}\right)}^{\frac{d-\gamma_\text{S}}{2}}. 
\end{multline}

\section{From a bound to an ansatz}
In the previous appendix we derived a bound on the logical error rate.
The form of the r.h.s.\@ expression of \autoref{eq:analytical bound} is based on the fact that the bound we used on the number of SAW scales as some constant to the power the length of the walk.
In practice,  $|\text{SAW}(l)|$ is asymptotically equivalent to an expression with such exponential scaling in $l$, which motivates the use of an ansatz of the same form as the bound we derived.

To simplify the expression, we will keep only the lowest order in $\frac{p_{\text{B}}}{p_{\text{B}}^*}$ and $\frac{p_{\text{S}}}{p_{\text{S}}^*}$ (considered to be of similar order): ${\left(\frac{p_{\text{B}}}{p_{\text{B}}^*}\right)}^{d/2}$ and ${\left(\frac{p_{\text{S}}}{p_{\text{S}}^*}\right)}^{d/2}$.
Also using the fact that each sum contains at most $\bigO{d^2}$ terms and that $\beta(p_{\text{B}}) = {(1+a)}^{-1} + \bigO{\sqrt{\frac{p_{\text{B}}}{p_{\text{B}}}^*}}$, we have:
\begin{multline*}
 \sum_{\gamma_\text{S}\geq d} {\left(\frac{p_{\text{S}}}{p_{\text{S}}^*}\right)}^{\gamma_\text{S}/2} + \sum_{\gamma_\text{B}\geq d} {\left(\frac{p_{\text{B}}}{p_{\text{B}}^*}\right)}^{\gamma_\text{B}/2} \\
+ \beta(p_{\text{B}}) \sum_{\gamma_\text{S}\geq1} {\left(\frac{p_{\text{S}}}{p_{\text{S}}^{**}}\right)}^{\gamma_\text{S}/2} {\left( \frac{p_{\text{B}}}{p_{\text{B}}^*}\right)}^{\frac{d-\gamma_\text{S}}{2}\bigr\rvert_{>0}}\\
\approx {\left(\frac{p_{\text{S}}}{p_{\text{S}}^*}\right)}^{d/2} +  {\left(\frac{p_{\text{B}}}{p_{\text{B}}^*}\right)}^{d/2} \\ + {(1+a)}^{-1}\sum_{1\leq\gamma_\text{S}\leq d} {\left(\frac{p_{\text{S}}}{p_{\text{S}}^{**}}\right)}^{\gamma_\text{S}/2} {\left( \frac{p_{\text{B}}}{p_{\text{B}}^*}\right)}^{\frac{d-\gamma_\text{S}}{2}}
\end{multline*}
where we have treated $p_{\text{S}}^{**}$ as if it was a constant.

To transform this bound into an ansatz, we take the possible polynomial factor in front of each term as some free real parameter.
This choice has been first motivated by the success of accordance between data and fit of the widely used ansatz $P_L = \alpha {\left(\frac{p}{p_{\text{th}}}\right)}^{\frac{d+1}{2}}$ (with $p$ the physical noise) in surface code memory experiments.
Indeed, the derivation done above is just an adaptation of the one done for a regular surface code in~\cite{Dennis_2002}.
We also replace $p_{\text{S}}$, $p_{\text{B}}$ by $p_{\text{Bulk}}$, $p$ which are the relevant noise strengths in our circuit-level simulations.
In the same spirit, we transform every combinatorial constant into a fit parameter, namely $\alpha_c$, $p^*$, $p_{\text{Bell}}^*$. Going to circuit-level noise model modifies the structure of the matching graph by adding extra diagonal edges (see \autoref{fig:matching graph}), and hence modifies the combinatorial constants.

In the end, the ansatz for the $X$ logical error rate per round is:
\begin{multline}\label{eq:anzatz}
     P_{L}^X \approx
    \alpha_1 {\left(\frac{p_{\text{Bell}}}{p_{\text{Bell}}^*}\right)}^{\frac{d+1}{2}} +  \alpha_2 {\left(\frac{p}{p^*}\right)}^{\frac{d+1}{2}} \\ 
    + \alpha_3 \sum_{1\leq i\leq d} {\left(\frac{p_{\text{Bell}}}{p_{\text{Bell}}^{*}}{\left(1+\frac{\alpha_c}{1-\sqrt{p/p^*}}\right)}^2\right)}^{\frac{i}{2}} {\left( \frac{p}{p^*}\right)}^{\frac{d+1-i}{2}}
\end{multline}
where the free fit parameters are $\alpha_1$, $\alpha_2$, $\alpha_3$, $\alpha_c$, $p^*$ and $p_{\text{Bell}}^*$.
In the same way as \autoref{new threshold}, we can define 
\begin{equation}\label{new threshold circuit level}
    p_{\text{Bell}}^{**}= p_{\text{Bell}}^* {\left(1+\frac{a}{1-\sqrt{\frac{p}{p^*}}}\right)}^{-2}
\end{equation}.

\begin{figure}
    \centering
    \includegraphics[width=\linewidth]{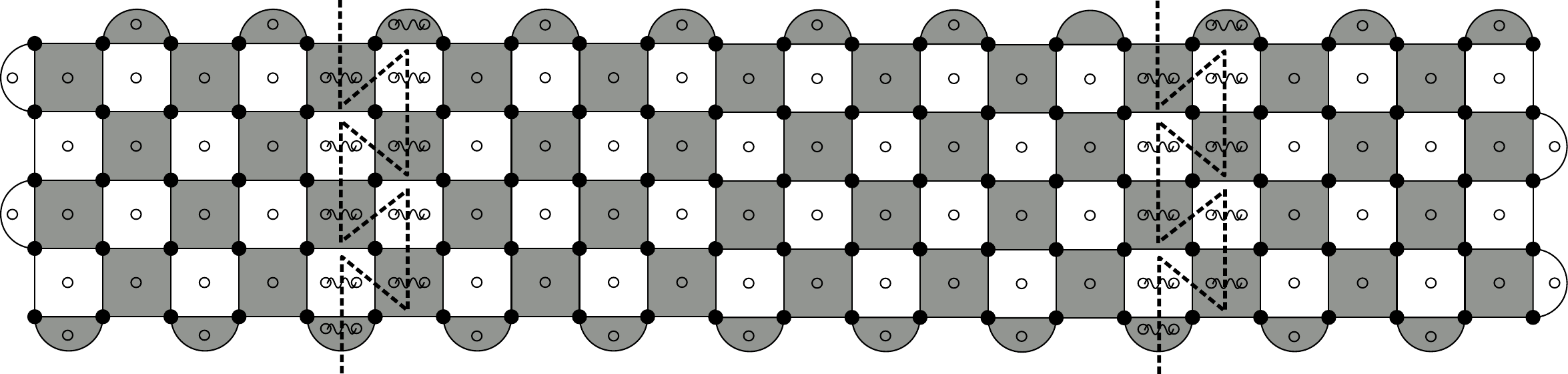}
    \caption{$d\times(4d+3)$ rotated surface code patch with one seam on the $d+1$-th column of data qubits and another one on the $3d+3$-th column.
    The distance is $d=5$ on the figure.
    This is reproducing the situation of a rectangular patch shared on 3 processors on the compact layout of \autoref{fig: compact layout}.
    Black dots are data qubits, circles are ancillae qubits and circles linked by a sinusoidal line are Bell pairs ancillae.
    The dashed zig-zag line represent the physical separation of the different QPUs.}
    \label{fig:patch multiseam}
\end{figure}

\section{Surface code patch across several QPUs}\label{appendix: multiple seam}

Merging lattice surgery operations on the architecture presented in \autoref{fig: compact layout} requires the use of a rectangular patch spanning across multiple QPUs in order to reach logical qubits sitting far apart from each other.

In this appendix, we assume the reader is familiar with the vocabulary and the path counting analogy defined in \autoref{annexe: ansatz}.
As pointed out in~\cite{ramette2023faulttolerantconnectionerrorcorrectedqubits}, when considering a patch with multiple ``seams'', we should rigorously take into account non-trivial closed paths that hops off from one seam and joins another.
However, under the bulk threshold, this term is exponentially suppressed with respect to the spacing between the two seams.

In our architecture, this spacing is typically the length of a processor which happens to be at least $2d+1$ (see \autoref{fig: compact layout}).
Therefore, as we typically consider operations at $p={10}^{-3}$, the probability associated with a non-trivial closed path leaving one seam and joining another seam would be smaller than that of minimal-length non-trivial closed path by a factor of ${\left(\frac{p}{p^*}\right)}^{\frac{2d+1-d}{2}}\approx {\left(\frac{1}{7}\right)}^{16} \approx 3 \times {10}^{-14}$ for $d=31$.
The $2d+1$ contribution in the exponent is the minimal length of an error chain between the two seams.
The $-d$ comes from the fact that on the rotated surface code matching graph, edges are $\SI{45}{\degree}$ tilted so that a minimal length chain between the two seams can also close the gap between the top and bottom boundaries (note that in this case, the seam is not used by the chain).
The exponent is divided by two as at most half of the edges can come from the recovery operation.
As this suppressing factor is of the order of magnitude of $P_L$ itself, one can safely consider each of the seams as being independent.

If we neglect the term of interaction between two seams, a rotated surface code patch of size $d_x\times d_z$ with $n_{\text{seam}}$ seams sufficiently far from each other, the residual logical error rate per round is approximated as:
\begin{multline}\label{eq:multiply split log er rate model}
    P_L^X \approx \alpha_1n_{\text{seam}} {\left(\frac{p_{\text{Bell}}}{p_{\text{Bell}}^*}\right)}^{\frac{d_x+1}{2}} + \alpha_2\frac{d_z}{d_x}{\left(\frac{p}{p^*}\right)}^{\frac{d_x+1}{2}} \\
    + \alpha_3 n_{\text{seam}} \sum_{1\leq i\leq d_x} {\left(\frac{p_{\text{Bell}}}{p_{\text{Bell}}^{**}}\right)}^{\frac{i}{2}} {\left( \frac{p}{p^*}\right)}^{\frac{d_x+1-i}{2}}
\end{multline}
Note that there is no additional free parameter compared with our ansatz for the one-seam case (\autoref{eq:anzatz}).

We simulated the two seam case to evaluate the pertinence of this ansatz.
On \autoref{fig:multiple seam}, we report the logical error rate of a $d\times(4d+3)$ patch with one seam on the $d+1$-th column of data qubits and another one on the $3d+3$-th column, as illustrated in \autoref{fig:patch multiseam}.
This simulation reproduces the $2d+1$ spacing between two seams on our architecture presented in \autoref{fig: compact layout}.
In \autoref{fig:multiple seam}, the dashed lines are the logical error rate predicted by the ansatz given by \autoref{eq:multiply split log er rate model} for $d_x=d$, $d_z=4d+3$ and $n_{\text{seam}}=2$.
The values $\alpha_1$, $\alpha_2$, $\alpha_3$, $\alpha_c$, $p^*$ and $p_{\text{Bell}}^{*}$ are the ones obtained from the fit of the one seam case.

\begin{figure}
    \centering
    \includegraphics[width=\linewidth]{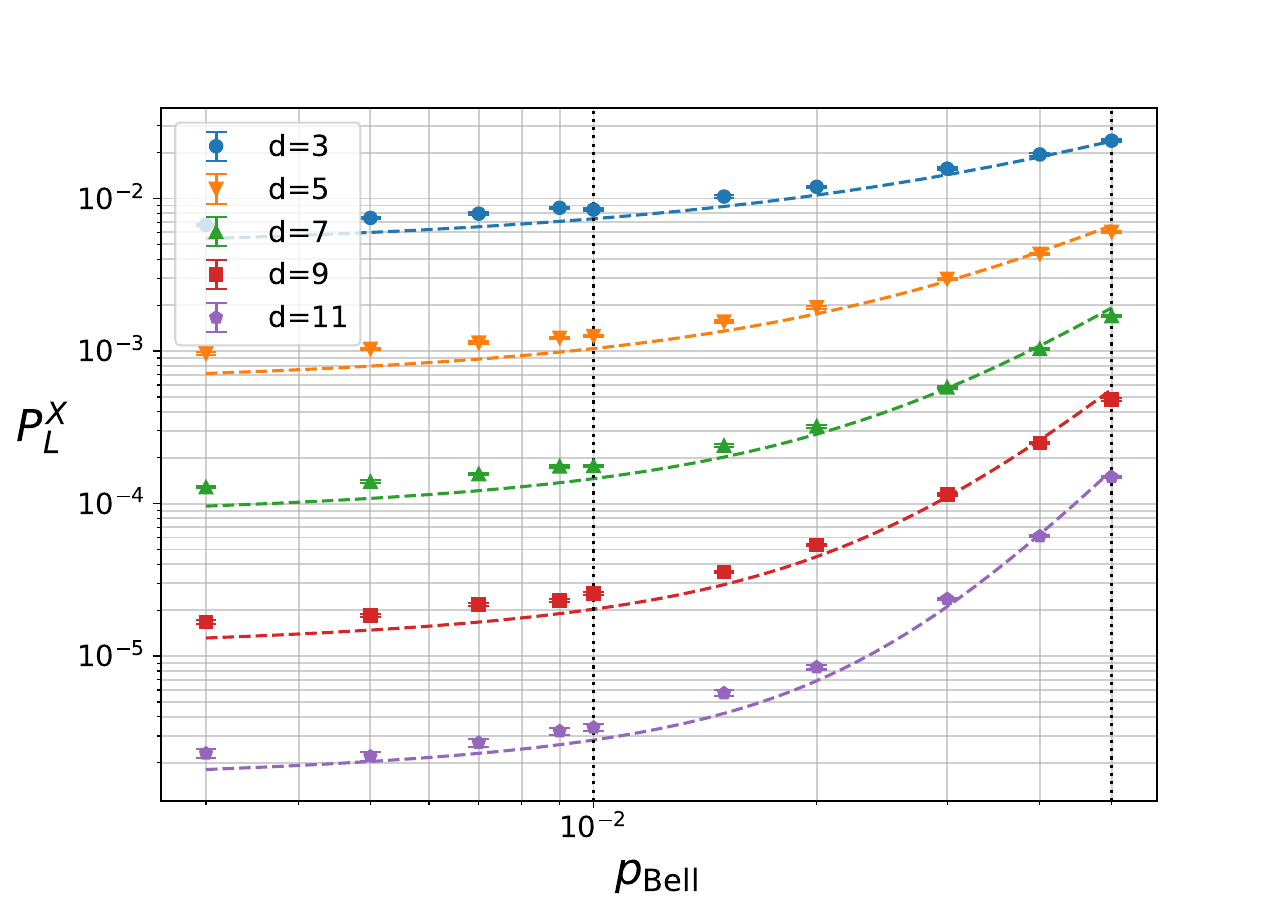}
    \caption{Logical error rate of a $\ket{0}_L$ memory experiment on a $d\times(4d+3)$ surface code with and two seams on data qubits columns $d+1$ and $3d+3$ with $p={10}^{-3}$ for various distances $d$.
    Note that the relevant regime for $p_{\text{Bell}}$ in our resource estimation is $\left[\num{1e-2},\num{5e-2}\right]$.
    Dashed lines represent the logical error rate model \autoref{eq:multiply split log er rate model} with parameters deduced from the fit on a single seam rotated surface code patch (see \autoref{annexe: fit methodology}).
    The two vertical dotted lines delimits the interval in which we use this model in the resource estimation.}
    \label{fig:multiple seam}
\end{figure}




%% file: appendix_numerical_methods.tex
\section{Noise model}\label{annexe:noise model}

We will refer to a surface code patch without any teleported measurement as a ``regular'' surface code as opposed to a ``distributed'' surface code for a surface with one seam (with two columns of teleported stabilizer measurements).

\subsection{Regular surface code noise}

In all of our regular surface code memory experiments, we use the fully depolarizing model of parameter $p$ which is summarized in this \autoref{tab:noise model}.
In our memory experiment, time is discretized into time-steps; during a time step a qubit can undergo a quantum gate such as $H$ or CNOT, a $Z$-basis measurement denoted M$Z$, a reset operation in state $\ket{0}$ denoted R or do nothing which is called ``idling''.
\begin{table}[h!]
\centering
\begin{tabular}{|c|c|c|c|c|}
\hline
 Operation & Idle, H   & CNOT      & M$Z$      & R \\ \hline
Noise channel  & DEPOL1($p$) & DEPOL2($p$) & ERRX($p$) & ERRX($p$) \\ \hline
\end{tabular}
\caption{Noise model for regular surface code memory experiments.
There is only one shared noise parameter $p$ for every noisy operations.}
\label{tab:noise model}
\end{table}

Each noise channel is characterized by its Kraus operators: 
\begin{itemize}
    \item DEPOL1($p$) is the one qubit depolarizing noise channel.
    It has the following set of Kraus operators:
    $\{ \sqrt{\frac{p}{3}} X,
    \sqrt{\frac{p}{3}} Y,
    \sqrt{\frac{p}{3}} Z,
    \sqrt{1-p} \mathbb{I} \}$.
    \item DEPOL2($p$) is the two qubits depolarizing noise channel.
    It has the following set of Kraus operators: $\sqrt{\frac{p}{15}}\{ X\otimes X,
 X\otimes \mathbb{I},
 X\otimes Y,
 X\otimes Z,
 Z\otimes X,
 Z\otimes \mathbb{I},
 Z\otimes Y,
 Z\otimes Z,
 Y\otimes X,
 Y\otimes \mathbb{I},
 Y\otimes Y,
 Y\otimes Z,
 \mathbb{I}\otimes X,
 \mathbb{I}\otimes Y,
 \mathbb{I}\otimes Z\}\cup\{
\sqrt{1-p} \mathbb{I}\otimes\mathbb{I}\}$.
Note that in the first set, we factorized a common prefactor of every operator.
    \item ERRX($p$) is a one qubit Pauli $X$ noise channel.
    Its Kraus operators are: $\{\sqrt{p}X, \sqrt{1-p}\mathbb{I}\}$
\end{itemize}
We've made the conservative choice of using only measurement and reset in $Z$ basis.
This implies adding two time steps for Hadamard gates in the syndrome extraction routine.
For simulation with distributed measurement, we use the extended noise model described in \autoref{Annexe noise Bell}.

\subsection{Bell pair preparation noise}\label{Annexe noise Bell}

The noise model that is used for Bell state preparation is a 2-qubit depolarizing noise channel of parameter $p_{\text{Bell}}$ happening right after the preparation of the Bell pair.
This noise channel is symmetric between the two qubits of the Bell pair, hence their order is arbitrary.
There are 15 different non-trivial Kraus operators in this noise channel.
As before, the set of Kraus operators is:
$\sqrt{\frac{p_{\text{Bell}}}{15}}\{ X\otimes X,
 X\otimes \mathbb{I},
 X\otimes Y,
 X\otimes Z,
 Z\otimes X,
 Z\otimes \mathbb{I},
 Z\otimes Y,
 Z\otimes Z,
 Y\otimes X,
 Y\otimes \mathbb{I},
 Y\otimes Y,
 Y\otimes Z,
 \mathbb{I}\otimes X,
 \mathbb{I}\otimes Y,
 \mathbb{I}\otimes Z\}\cup\{
\sqrt{1-p_{\text{Bell}}} \mathbb{I}\otimes\mathbb{I}\}$.
Using the stabilizer operators of the Bell pairs, $X\otimes X$, $Y\otimes Y$ and $Z\otimes Z$, we can equivalently map the above set of Kraus operators on the set $\{\sqrt{1-\frac{3p_{\text{Bell}}}{4}}\mathbb{I}\otimes\mathbb{I}, \sqrt{\frac{4p_{\text{Bell}}}{15}}X\otimes \mathbb{I}, \sqrt{\frac{4p_{\text{Bell}}}{15}}Z \otimes \mathbb{I}, \sqrt{\frac{4p_{\text{Bell}}}{15}}Y \otimes \mathbb{I}\}$.
Therefore, each error in the set $\{X\otimes \mathbbm{1}, Z \otimes \mathbbm{1}, Y \otimes \mathbbm{1}\}$ is associated with a probability $\frac{4}{15} p_{\text{Bell}}$ to happen.
From there, one can notice that without loss of generality, thanks to Bell pair stabilizers, $X$, $Z$ and $Y$ errors can always be considered as happening on the branch of the Bell pair involved in only one CNOT.
Given a teleported $X$ stabilizer measurement, if a $X\otimes\mathbb{I}$ error is triggered it will propagate as a hook error through the CNOT to the data qubit on the seam (the data qubits column between the two columns of Bell pairs ancillae,  see \autoref{fig:seam}).
However, if a $Z\otimes\mathbb{I}$ happens it won't propagate through the CNOT and will only flip the stabilizer outcome.

When investigating logical $X$ errors, the situation is equivalent to increasing the $X$ physical error rate by $\frac{8}{15} p_{\text{Bell}}$ for every data qubits on the seam at each syndrome extraction.
Those $X$ errors are detected by teleported $Z$ stabilizers which also have an augmented probability of returning a wrong outcome.
When turning off the background noise $p$ on every operation, the  Bell pair preparation noise can almost be seen as a phenomenological repetition code on top of the surface code.

The first difference with a phenomenological noise model is that  $Y\otimes\mathbb{I}$ acts as a correlated data qubit and stabilizer measurement error as it flips both.
Secondly, due to the CNOTs in the extraction circuit, $X\otimes\mathbb{I}$ on a $Z$ stabilizer ancilla Bell pair can not only flip the data qubit but also propagate to neighboring teleported $Z$ stabilizer and flip it, once more introducing correlation between data qubit noise and measurement noise.

It is also worth mentioning that for preparing the Bell pair as a Stim~\cite{gidney2021stim} circuit, a $H\otimes\mathbb{I}$ gate followed by a CNOT gate are used in a single time step (tick), and the described error channel is only applied afterwards.
This has been done with the idea that this paper is not considering Bell state generation protocols but only its output fidelity.
If the generation protocol appears to be slower, one could either use multiplexing or work with the fact that it would add idle errors on data qubits.
Indeed, regular ancilla qubit resets can be delayed.
Depending on the additional idling time, this would slightly modify the fit parameters.

\begin{figure}
\begin{center}
    \begin{minipage}{0.23\textwidth}
        \centering
        \includegraphics[width=\textwidth]{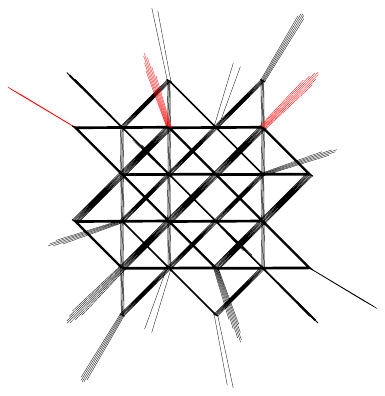}
    \end{minipage}
    \hfill
    \begin{minipage}{0.23\textwidth}
        \centering
        \includegraphics[width=\textwidth]{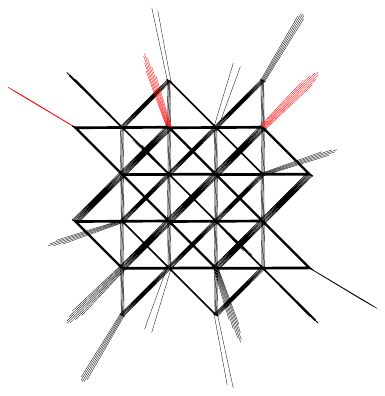}
    \end{minipage}
    \caption{Circuit-level matching graph of a distance $d=5$ rotated surface code memory experiment generated by stim package~\cite{gidney2021stim} with (on the left) and without (on the right) teleported measurements.
    Time axis is orthogonal to the plan of the figure.
    We notice the disappearance of a set of time-diagonal edges between Bell-pair based checks represented the third vertical line of the graph projection.}\label{fig:matching graph}
\end{center}
\end{figure}

As an additional remark, we noticed a few differences between the matching graphs of distributed and regular surface code.
As shown in \autoref{fig:matching graph}, some diagonal hook error edges no longer exist.
This is due to well-identified differences in the error propagation of the two syndrome extraction schemes.
This feature may simplify the decoding process although we didn't investigate enough to draw conclusions on this side.

\subsection{Different noise models}
One could use our code with a different noise model, closer to specific hardware as the SI1000 model used by Google in~\cite{Gidney_honeycomb}.
This can modify substantially the fitted parameters values.
For example with a model closer to the SI1000 model, one recovers a surface code threshold of $p^*=1\%$.
Using such a threshold reduces our monolythic resource estimation to 25 millions qubits (going from distance 35 to 31) getting closer to the space cost reported in~\cite{Gidney_2021}.

\section{Fit methodology}\label{annexe: fit methodology}
The Monte Carlo samplings of our memory experiment circuits are done with the Stim and Sinter python modules~\cite{gidney2021stim}, with PyMatching as decoder~\cite{higgott2021pymatchingpythonpackagedecoding}.
The fits in this paper are done using a least-squares method via the Levenberg-Marquardt algorithm.

\subsection{Fitting procedure for a distributed surface code patch memory experiment}

Here we detail the fitting procedure used for the case of a memory experiment of a surface code with two columns of teleported stabilizer measurements as displayed on \autoref{fig:seam}. We refer to such a patch as a ``distributed'' surface code patch.

Let $f_d(p,p_{\text{Bell}};\{\theta\})$ be our logical error rate per round ansatz where $\{\theta\}$ is the set of parameters to fit.
For our logical error rate model of a distributed surface code $\ket{0}_L$ memory experiment, $\{\theta\}=\{\alpha_1,\alpha_2,\alpha_3,\alpha_c, p^*,p_{\text{Bell}}*\}$ and the ansatz comes from \autoref{eq:anzatz}:
\begin{multline*}
f_d(p,p_{\text{Bell}};\{\theta\})
    = \alpha_1 {\left(\frac{p_{\text{Bell}}}{p_{\text{Bell}}^*}\right)}^{\frac{d+1}{2}} +  \alpha_2 {\left(\frac{p}{p^*}\right)}^{\frac{d+1}{2}} \\ 
    + \alpha_3 \sum_{1\leq i\leq d}{\left(\frac{p_{\text{Bell}}}{p_{\text{Bell}}^*}{\left[{1+\frac{\alpha_c}{1-\sqrt{\frac{p}{p^*}}}}\right]}^2\right)}^{\frac{i}{2}} {\left( \frac{p}{p^*}\right)}^{\frac{d+1-i}{2}}.
\end{multline*}
For each triplet $(d,p,p_{\text{Bell}}),$ we build an estimator of the logical error rate via a Monte-Carlo sampling of a $\ket{0}_L$ memory experiment over $k=3d$ rounds.
This estimator is built as $\hat{p}_{L,k}(d,p,p_{\text{Bell}})=\frac{\#\text{logical errors}}{\#\text{ shots}}$ from which we deduce the logical error rate per round $\hat{p}_L(d,p,p_{\text{Bell}})=1-{(1-\hat{p}_{L,k})}^{\frac{1}{k}}$.
One can compute the associated statistical error through a Taylor expansion of this expression:  
\begin{equation*}
\sigma_{\hat{p}_L}=\sqrt{\frac{\hat{p}_{L,k}(1-\hat{p}_{L,k})}{n}}/\left(k {(1-\hat{p}_{L,k})}^{1-1/k}\right)
\end{equation*}
where $n$ is the number of shots.
We then minimize the following quantity: 
\[
\chi^2(\{\theta\}) = \sum_{(d,p,p_{\text{Bell}})} {\left(\frac{\hat{p}_L-f_d(p,p_{\text{Bell}};\{\theta\})}{\sigma_{\hat{p}_L}}\right)}^2
\]
\begin{figure}
    \centering
    \includegraphics[width=\linewidth]{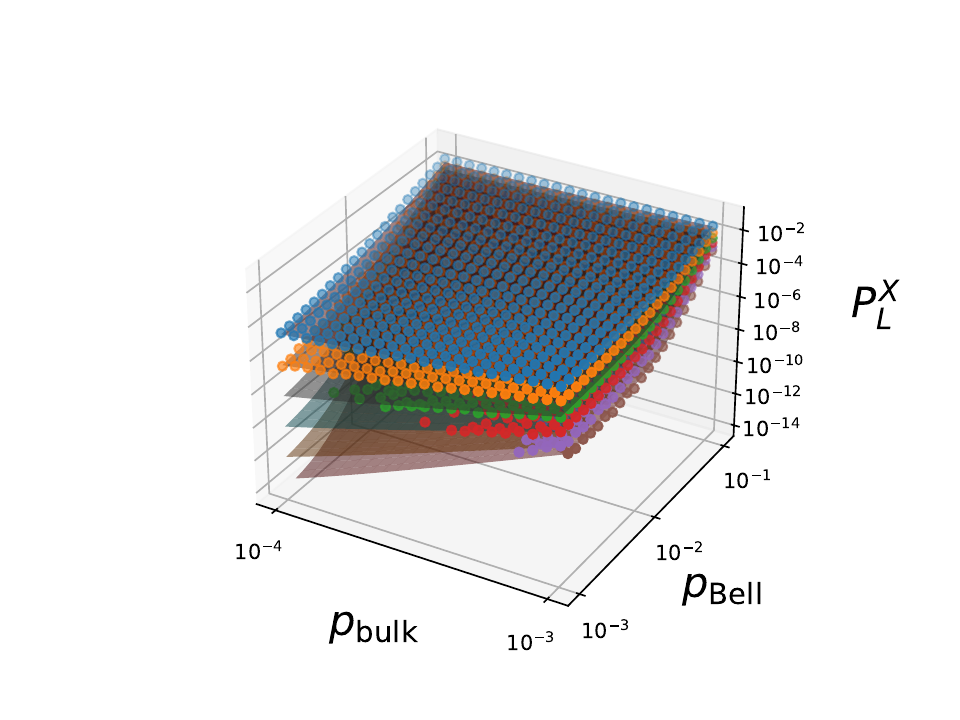}
    \caption{Fitted error model in the $(p_{\text{Bell}},p,p_L)$ space. Surfaces are the fitted model and each point represent a Monte Carlo sampling.
    We filtered the data to show only points where the standard deviation is less than half the value of the logical error.
    Colors refers to distances 3, 5, 7, 9, 11, 13.}
    \label{fig:3d fit}
\end{figure}

As the ansatz has been built in the regime of $\frac{p}{p^*}, \frac{p_{\text{Bell}}}{p_{\text{Bell}}^*}<1$ the fit is done on a grid of points $(p,p_{\text{Bell}})$ evenly distributed on a logarithmic scale in the space $\left[{10}^{-4},{10}^{-3}\right]\times\left[{10}^{-3}, 5 \times {10}^{-2}\right]$ which corresponds to working in the regime of $\frac{p}{p^*}, \frac{p_{\text{Bell}}}{p_{\text{Bell}}^*}<\frac{1}{6}$.
For each of these points, we do the sampling for the set of distances $d\in \{3,5,7,9,11,13\}$ but only use the $d\in \{5,7,9,11,13\}$ points for the fit.
It is a common behavior in the literature not to use small distances in the fit of asymptotic ansatz~\cite{Stephens_2014}, because the approximation that we did for path counting in \autoref{annexe: ansatz} becomes inaccurate on small lattices.
To ensure quality of the data points used in the fit, we also filter and use only points with $\sigma_{\hat{p}_L}< \frac{\hat{p}_L}{2}$.

Above $p_{\text{Bell}}>5 \times {10}^{-2}$, it becomes harder to do a meaningful fit with such an asymptotic-based ansatz.
This sets a limit to the use of the logical error rate ansatz that we propose.
However, as the resource estimation suggests, if one would need to work with such noisy Bell pairs, one would better distillate them beforehand.
Otherwise, the distance of the code would blow up in order to achieve a logical error rate below ${10}^{-10}$.

The result of this fitting procedure is shown in \autoref{fig:3d fit}. Note that we display data points up until $p_{\text{Bell}}=0$ even though they are not used in the fit, as detailed above.

\subsection{Goodness of the fit}

We used two different methods to estimate the uncertainties over fit parameters.
The first one is the default method from the lmfit package~\cite{lmfit}, relying on the inverted Hessian matrix returned by \lstinline{scipy.optimize.leastsq}.
This corresponds to the fit report displayed in \autoref{tab:fit_params_leastsq}.
However, as this method may not exhibit meaningful results for strongly correlated variables, we also used an F-test implemented in~\cite{lmfit} which provides more robust confidence intervals that we report in \autoref{tab:conf_intervals_2}.

\begin{table}[h!]
\centering
\begin{tabular}{|c|c|}
\hline
\textbf{Parameter} & \textbf{Value with Uncertainty} \\ \hline
$\alpha_c$ & $0.2057 \pm 0.0157$ \\ \hline
$\alpha_3$ & $0.05326 \pm 0.00108$ \\ \hline
$\alpha_2$ & $0.04507 \pm 0.00108$ \\ \hline
$\alpha_1$ & $0.09789 \pm 0.01499$ \\ \hline
$p^*$ & $0.007176 \pm 0.000039$ \\ \hline
$p_{\text{Bell}}^*$ & $0.2983 \pm 0.0110$ \\ \hline
\end{tabular}
\caption{Fit parameters with uncertainties computed by inversing the Hessian matrix returned by \lstinline{scipy.optimize.leastsq}.}\label{tab:fit_params_leastsq}
\end{table}

\begin{table*}[t]
\centering
\begin{tabular}{|c|c|c|c|c|}
\hline
\textbf{Parameter} & \textbf{\SI[detect-all=true]{68.27}{\percent}} & \textbf{\SI[detect-all=true]{95.45}{\percent}} & \textbf{\SI[detect-all=true]{99.73}{\percent}} & \textbf{Best Estimate} \\ \hline
$\alpha_c$ & $+0.0162, -0.0162$ & $+0.0331, -0.0331$ & $+0.0554, -0.0554$ & $0.2057$ \\ \hline
$\alpha_3$ & $+0.0011, -0.0011$ & $+0.0022, -0.0022$ & $+0.0038, -0.0038$ & $0.05326$ \\ \hline
$\alpha_2$ & $+0.0011, -0.0011$ & $+0.0022, -0.0022$ & $+0.0034, -0.0034$ & $0.04507$ \\ \hline
$\alpha_1$ & $+0.0172, -0.0172$ & $+0.0363, -0.0363$ & $+0.0591, -0.0591$ & $0.09789$ \\ \hline
$p^*$ & $+0.000040, -0.000040$ & $+0.00010, -0.00010$ & $+0.00012, -0.00012$ & $0.007176$ \\ \hline
$p_{\text{Bell}^*}$ & $+0.0104, -0.0104$ & $+0.0215, -0.0215$ & $+0.0332, -0.0332$ & $0.2983$ \\ \hline
\end{tabular}
\caption{Confidence intervals for fit parameters from a F-test computed by the \lstinline{conf\_interval} method of lmfit~\cite{lmfit}.}\label{tab:conf_intervals_2}
\end{table*}

%% file: appendix_architecture.tex
\begin{figure*}
    \centering
    \includegraphics[width=0.8\linewidth]{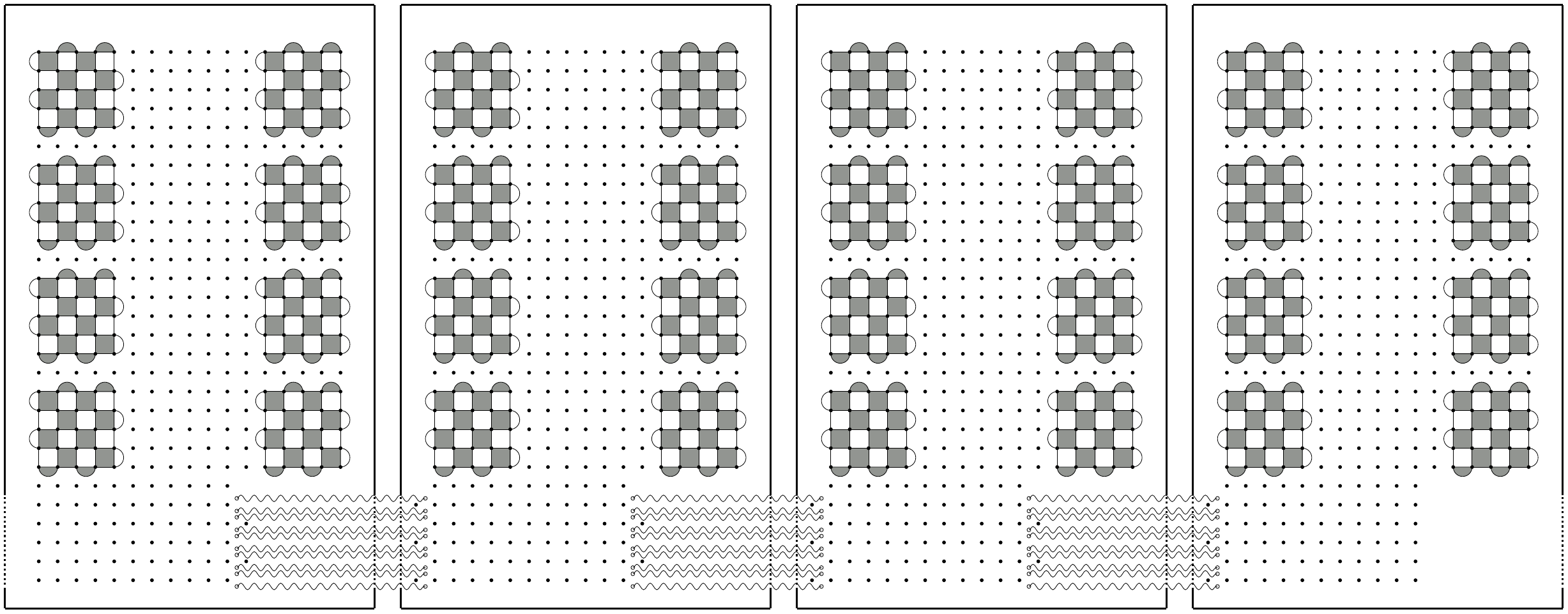}
    \caption{Chain architecture of processors with the studied compact layout, here represented with processors with $n_{\text{rows}}=4$.
    Dots are data qubits. For readability, ancillae qubits have been omitted except for Bell pairs represented by two circle linked by a sinusoidal line.
    The different patches exhibit the location of logical qubits in the layout.}
    \label{fig: compact layout}
\end{figure*}

\section{Compact layout}\label{Appendix: layout}

In this section, we present a layout that we developed and used in our resource estimation for our distributed architecture.
This layout, displayed in \autoref{fig: compact layout}, is inspired by the compact layout proposed by Litinski in Ref.~\cite{Litinski_2019}.

A processor with $n_{\text{rows}}$ rows of distance $d$ surface code qubits will contain $N_{\text{proc.\@ log}}=2n_{\text{rows}}$ logical qubits and requires $N_{\text{data}}$ physical data qubits with $N_{\text{data}}$ being 
\[
N_{\text{data}}= \left(3d+2\right)\left((d+1)n_{\text{rows}} -1\right) +(2d+1)(d+1) +d
\]
The first term corresponds to the rectangle area with the logical qubits and the in-between routing area.
The second term is the routing area at the bottom of the processor on \autoref{fig: compact layout}.
Finally, the $+d$ takes into account the additional data qubits required to form the column of data qubits in between the two columns of Bell pair ancillae ($+d$ accounts for the data qubits needed for the links with previous and next processors on the chain). 
Then taking into account syndrome qubits as well as the additional qubits used to form the Bell pair links, the total number of physical qubits per processor is
\[
N_{\text{phys.\@ proc.\@}}= 2N_{\text{data}}-1 +2d
\]
Note that the $2d$ additional ancillae qubits correspond to the doubling of two columns of ancillae to later create Bell pairs.
We call $n_{\text{proc.\@}}$ the number of processors in the distributed setting.
It is fixed by the total number of logical qubits required and by $n_{\text{rows}}$ which is chosen to be the smallest possible so that a full factory of $N_{\text{factory}}$ physical qubits (see \autoref{Annexe: Distillation}) fits in a processor.


In this layout, we don't have direct access to all logical Pauli observables so that, as detailed in \autoref{Annexe lattice surgery}, the CNOT gate is performed in $4dt_c$ (instead of $2dt_c$ as it would be the case on a layout with bigger routing areas), with $t_c$ the cycle time of the surface code.

Note that we may think about a better inter-processor connectivity~\cite{guinn2023codesignedsuperconductingarchitecturelattice}.
One could typically use several links between nearest-neighbors processors in some specific area of the architecture to enhance the parallelization.
It is also possible to add connectivity, for example by having nearest neighbor connectivity on a 2D grid of QPUs.
Interestingly, as highlighted in \autoref{appendix: multiple seam} splitting a patch over multiple QPUs doesn't introduce a prohibitive logical noise cost.
More precisely, each QPU interface that the patch crosses adds the same contribution to the logical error rate so that when manipulating simultaneously several merged patches, only the total number of interfaces used as well as the sum of the lengths of the patches matter.
Compared with the exponential suppression of noise given by error-correction, the linear increase of noise with the individual lengths and number of interfaces on each rectangular patch don't matter.
As a result, enhancing parallelization would certainly be a better compilation criterion than minimizing the number of crossed seams during individual logical gates.
One could also imagine switches of quantum links between a central processing unit and several chips dedicated to storage, mimicking the behavior of a memory~\cite{Gouzien_2021}.

Ultimately, to run large scale algorithms one would use dynamical layouts such as the one proposed in~\cite{gidney2019flexiblelayoutsurfacecode}.
Dynamical layouts save a lot of routing qubits by moving the working area accordingly during the execution of the algorithm.
Such layouts are algorithm dependent, so we made the choice to use a static layout in order to draw conclusions that are transposable to any quantum algorithm.

%% file: appendix_lattice_surgery.tex
\section{Logical CNOT}\label{Annexe lattice surgery}
\subsection{Lattice surgery on a distributed architecture}
In addition to single-qubit gates, lattice surgery primitives complete the set of logical operations on the surface code to achieve universality.
Those primitives, merging and splitting operations, offer a way to measure logical multi-qubit Pauli operators~\cite{Horsman_2012, fowler2019lowoverheadquantumcomputation}. In particular, the CNOT gate can be decomposed into multi-qubit Pauli measurements, as presented in \autoref{fig:cnot circuit}.
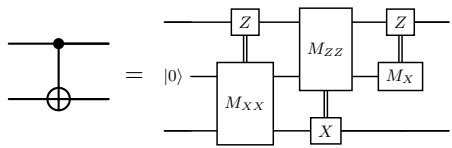
\begin{figure}[!h]
\begin{quantikz}
\qw&\ctrl{1}&\qw\\
\qw&\targ{}&\qw
\end{quantikz}
=
\scalebox{0.7}{
\begin{quantikz}
		  &\qw                               &\gate{Z}\vcw{1}		&\gate[wires=2]{M_{ZZ}}&\gate{Z}\vcw{1}&\qw \\
        &\lstick{$\ket{0}$}\wireoverride{n}&\gate[wires=2]{M_{XX}}&                      &\gate{M_X}\\
		  &\qw                               &                      &\gate{X}\vcw{-1}      &\qw 		    &\qw 
\end{quantikz}
}
\caption{CNOT gate circuit via lattice surgery from~\cite{fowler2019lowoverheadquantumcomputation}.}\label{fig:cnot circuit}
\end{figure}

To describe the topology and schedule of the lattice surgery operations, we use 3D topological diagrams; for a good introduction to spacetime diagrams see~\cite{Tan_2024}.
Typically, on \autoref{fig:spacetime cnot}, the horizontal dimensions represent the positions on the chip which is why we display the architecture layout below the spacetime diagram so that by projecting the pipes in this plane one can see where the patches stand at a given timestep.
The time axis is orthogonal to this plane.
Thickness represents either the patch width or length or the number of rounds of error correction.
Vertical grey (resp.\@ white) surfaces are $X$ (respectively $Z$) boundaries, i.e.\@ $X$ (or $Z$) two-qubits stabilizers.
Horizontal grey (resp.\@ white) tiles are transversal data qubits $X$ (resp.\@ $Z$) preparation or measurement; they typically last 1 round but we don't represent their thickness.

\begin{figure}
    \centering
    \includegraphics[width=\linewidth]{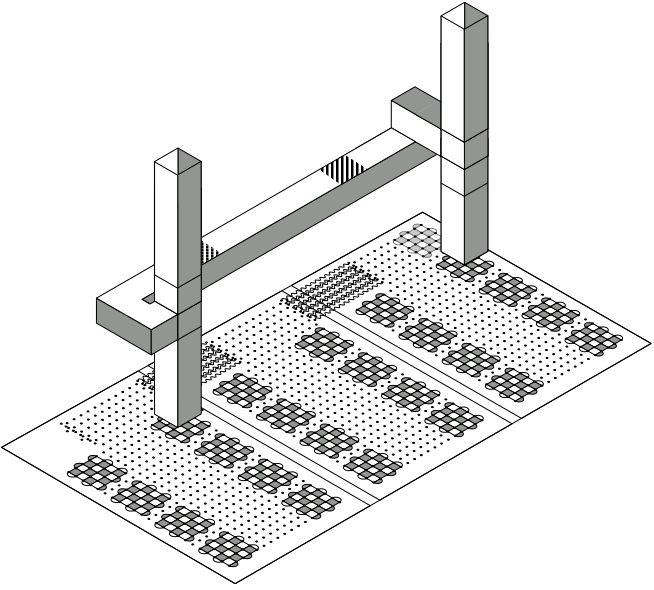}
    \caption{Spacetime diagram of a CNOT via lattice surgery on a distributed surface code architecture.
    Pipes represents surface code patches that are preserved through time.
    For readability, we omit to represent the straight vertical pipes of every logical idling qubits.
    The control qubit is the vertical pipe on the right and the target is the one on the left.
    Time is going vertically and horizontal slice represents the surface code patches existing on the layout at a given time.
    Vertical (resp.\@ horizontal) tiles are space (resp.\@ time) boundaries of $Z$ or $X$ type depending on the color of the tile (resp.\@ white and grey).
    One can notice that both control and target qubits vertical pipes doesn't have an horizontal tile; this is because those qubits are not measured at the end of the computation.
    We added long vertical pipes for control and target before and after the actual CNOT process to emphasize the fact that those qubits surface code patches already exist before and continue to exist after the CNOT.
    Dashed area on the pipes represent teleported stabilizer measurement at the interface between two processors.
    We represented the patch of the auxiliary logical qubit as blurred on the layout, it only exists during the two lattice surgery operations.}\label{fig:spacetime cnot}
\end{figure}

The representation of the CNOT is given in \autoref{fig:spacetime cnot}.
It follows the protocol explained in~\cite{Tan_2024}, taking additionally into account interfaces between QPU through Bell pairs.
A zoom on a patch crossing an interface is shown in \autoref{fig: patch at interface}.
\begin{figure}
    \centering
    \includegraphics[width=0.5\linewidth]{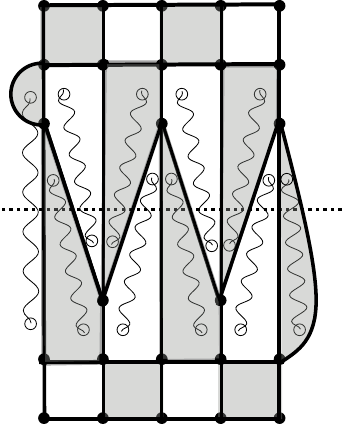}
    \caption{Zoom on the stabilizers of a distributed rotated surface code crossing an interface between two computing units (see \autoref{fig: compact layout}).
    The dotted line represent the physical separation between the two processors.
    Note that the second qubit of the Bell pair on the very left is unused.
    In practice only $2d-1$ Bell pairs per interface would be required (with $d$ the distance of the logical qubits in the algorithm).
    In practice we did the simulation with $2d$ Bell pairs because it symmetrizes the noise model which facilitates the ansatz construction.}
    \label{fig: patch at interface}
\end{figure}
The two vertical pipes are the two logical qubits whose patches always exist through the CNOT process.
Horizontal slices of pipes represent a surface code patch, its location on the layout can be deduced by doing a vertical projection on the layout at the bottom of the figure.
The stabilizers of the patches are measured $d$ times to ensure fault tolerance~\cite{gidney2019flexiblelayoutsurfacecode, Chamberland_2022}.

On \autoref{fig:spacetime cnot}, the auxiliary qubit in $\ket{0}$ of \autoref{fig:cnot circuit} is prepared transversally in a single time step right before the merging operation, represented by the horizontal pipe crossing several seams, of the $M_{XX}$ measurement.
It is then transversally measured in a single time step right after the splitting operation of the $M_{ZZ}$ measurement.
Note that a proper fault-tolerant preparation of $\ket{0}$ normally requires $d$ rounds of measurements, but here the ones of the following merging operation are sufficient.
Therefore, the terminology ``auxiliary logical qubit'' is quite artificial and only comes from the circuit picture (\autoref{fig:cnot circuit}); one could simply see it as a vertical pipe essential to link the two T-shaped spacetime pipe junctions while maintaining consistent pipe's sides coloring.

Lattice surgery methods to measure multi-qubit Pauli involving $Y$ operators are more involved, and not always needed~\cite{Chamberland_2022, Chamberland_2022_twist,litinski_lattice_2018, gidney_inplace_2024}.
Nonetheless, the cost of the Shor algorithm is dominated by CNOT and Toffoli~\cite{Gouzien_2021,Gouzien_2023} gates which only require multi-qubit $X_L$ and multi-qubit $Z_L$ measurements.

For table lookup circuit~\cite{Babbush_2018,Gouzien_2023}, one can reduce the number of logical gates to perform by using C$X...X$ gates instead of several CNOT gates.
Lattice surgery provides a natural way to apply those gates in the same manner as \autoref{fig:cnot circuit}, using the circuit given in \autoref{fig:cnotnot circuit}.

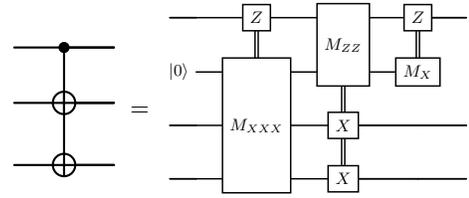
\begin{figure}
\begin{quantikz}
\qw&\ctrl{2}&\qw\\
\qw&\targ{}&\qw\\
\qw&\targ{}&\qw
\end{quantikz}
 = 
\scalebox{0.7}{
\begin{quantikz}
		  &\qw&\gate{Z}\vcw{1}               &\gate[wires=2]{M_{ZZ}}  &\gate{Z}\vcw{1}&\qw \\
        &\lstick{$\ket{0}$}\wireoverride{n}&\gate[wires=3]{M_{XXX}}&               &\gate{M_X}\\
		  &\qw&                              &\gate{X}\vcw{-1}        &\qw 		      &\qw \\
        &\qw&                              &\gate{X}\vcw{-1}        &\qw 		    &\qw \\
\end{quantikz}
}
\caption{C$XX$ logical gate circuit via lattice surgery from~\cite{fowler2019lowoverheadquantumcomputation}.}\label{fig:cnotnot circuit}
\end{figure}

Such $M_{X^k}$ measurements can be achieved on a distributed architecture in a similar way as presented in \autoref{fig:spacetime cnot}.
The only difference is that there would be more logical qubits involved in the computation.
Every target's vertical pipes are connected to the first horizontal pipe, see \autoref{fig:spacetime cnotnot}.
\begin{figure}
    \centering
    \includegraphics[width=\linewidth]{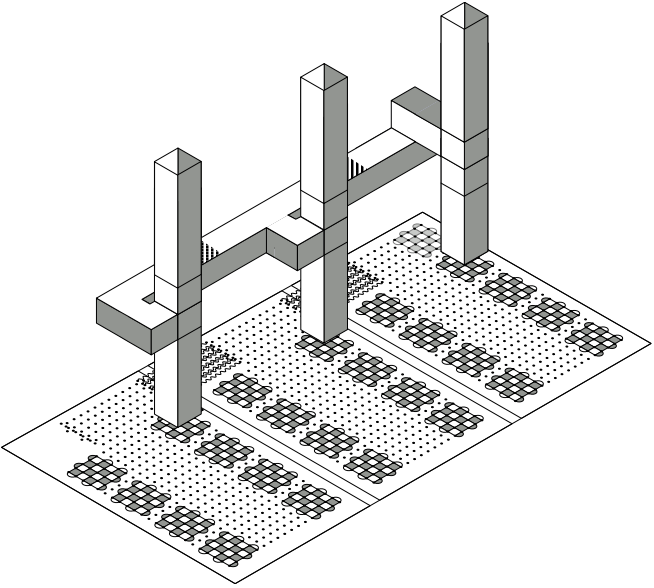}
    \caption{Spacetime diagram of a C$XX$ gate via lattice surgery on a distributed surface code architecture.
    Control is on the right, targets are in the center and on the left.
    See \autoref{fig:spacetime cnot}'s caption for reading guidelines.
    This spacetime diagram implements the circuit of \autoref{fig:cnotnot circuit}.
    We represented the patch of the auxiliary logical qubit as blurred on the layout, it only exists during the two lattice surgery operations.
    For more details about multi-qubit Pauli measurement via lattice surgery see~\cite{Litinski_2019,fowler2019lowoverheadquantumcomputation}.}
    \label{fig:spacetime cnotnot}
\end{figure}

\subsection{Logical CNOT on compact layout}

On \autoref{fig:spacetime cnot} and \autoref{fig:spacetime cnotnot} we have represented a convenient situation where logical Paulis were accessible without moving any qubit in the process.
However, one can't achieve such a situation with any locations of control and target logical qubits on the layout.
Specifically, we can't directly access $Z_L$ representative of the logical qubits in our architecture (\autoref{fig: compact layout}) and we must move patches in between the $M_{XX}$ and the $M_{ZZ}$ measurements.

The systematic way to do a CNOT gate in the compact layout of \autoref{fig: compact layout} is the following:
\begin{enumerate}
    \item Initialize an auxiliary qubit on the bottom left corner of the same chip as the controlled qubit.
    \item Do the $M_{XX}$ measurement with a routing patch, eventually crossing all the architecture, that we preserve for $d$ rounds.
    \item Move both the control and the auxiliary qubits in the central corridor of the chip (it takes $d$ rounds).
    \item Merge them through a vertical rectangular patch, during $d$ rounds.
    \item Measure transversally the logical auxiliary qubit in the $X$ basis.
    \item Move back the control qubit to its initial location.
\end{enumerate}
The whole process takes $4d$ rounds of syndrome extraction (neglecting the transversal preparations or measurements). Note that this process \emph{is not} the one illustrated on \autoref{fig:spacetime cnot} and \autoref{fig:spacetime cnotnot}, which we have chosen to illustrate the spacetime diagram for a convenient control and target location (without the need to move patches).
For the interested reader, we also give the spacetime diagram of the general case in \autoref{fig:spacetime cnot general}.
Note that the target qubit is only involved in a $M_{XX}$ measurement and our layout already provide access to the $X_L$ logical operator of every logical qubits so it doesn’t need to be moved.
\begin{figure}
    \centering
    \includegraphics[width=\linewidth]{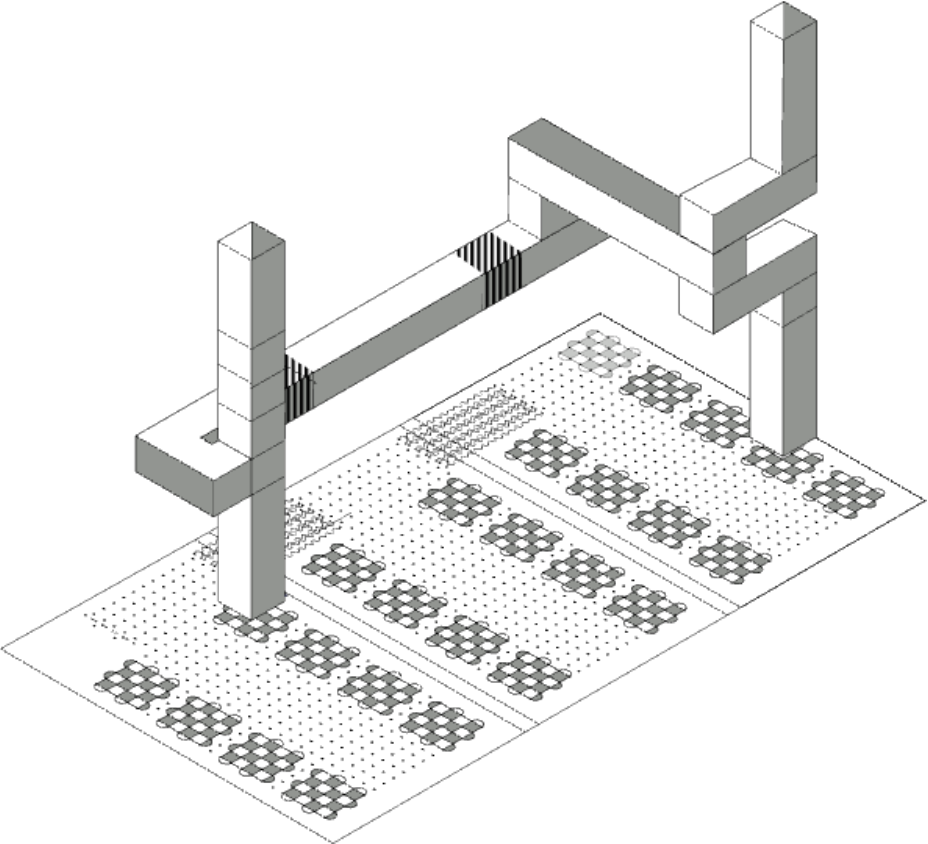}
    \caption{Spacetime diagram of a CNOT via lattice surgery in the general case when one needs to move the control and auxiliary qubit (whose initial position is represented with a shaded patch) to access their logical $Z$ operators.
    See \autoref{fig:spacetime cnot} for spacetime diagram reading guidelines.
    }
    
    \label{fig:spacetime cnot general}
\end{figure}

Actually, most of the CNOT could be done in $2d$ cycles as steps 2 and 3 could, in most situations, be done at the same time.
This is possible by initializing the auxiliary logical qubit in the central corridor next to the target qubit (instead of the bottom left of the control qubit's chip) thus eventually using a split surface code patch for the $M_{ZZ}$ measurement (instead of the $M_{XX}$).
However, this wouldn't be possible for every possible CNOT gates.
For example, dealing with the case of target and control qubits facing each other in the same row of the same chip would stick the control qubit during the $M_{XX}$ measurement making it impossible to perform the gate in $2d$ cycles.
One could eventually argue that avoiding such situations could be a task of the classical software compiler.
We preferred working with conservative assumptions and dealing with the scheme described above that works in any situation despite introducing an overall $1/2$ time overhead to the computation.

\subsection{CNOT logical error model}

In most of the Shor algorithm resource estimation based on lattice surgery~\cite{Gidney_2021, Gouzien_2023}, the logical error model of the CNOT is deduced from the logical error model of a memory experiment applied on $d\times d$ patches of the involved qubits (control, target and auxiliary) during the $2d$ error correction cycles of the CNOT gate.
This model forgets about all the errors that may happen in the horizontal pipes on \autoref{fig:spacetime cnot}.
Applying such a model in our situation wouldn't make sense because seams are only present in the merged patches represented by the horizontal pipes.
Forgetting about them in the error model would hide the effect of the noisy Bell state preparations.

The model that we used for a logical CNOT is the following:
\begin{itemize}
    \item We apply the regular surface code logical error model of a memory experiment from \autoref{eq:log error rate model} to all of the logical qubits (the number of logical qubits is called $N_{\text{log.}}$) during the $4d$ cycles of the CNOT.
    Taking into account logical idle noise in the same time.
    The associated probability of failure per round is:
    \[\mathbb{P}_{\text{log.\@ qubits}} =1-{\left(1-2\alpha{\left(\frac{p}{p_{\text{th}}}\right)}^{\frac{d+1}{2}}\right)}^{N_{\text{log.}}}\]
    where the factor $2$ is here to take into account both $X$ and $Z$ logical errors.
    Note that $p_{\text{th}}$ and $p^*$ are not exactly the same quantities.
    The first is the threshold in the regular surface code \autoref{eq:log error rate model} while the other is the bulk threshold in the distributed surface code \autoref{eq:split log er rate model}.

    \item We also use the memory experiment logical error model of a rectangular surface code patch, crossing all of the $n_{\text{proc.\@}}$ processors, whose error rate is given by \autoref{eq:multiply split log er rate model}, described in \autoref{appendix: multiple seam} during $d$ cycles.
    This model is conservative as in practice CNOT won't always be applied between two qubits that far apart from each others.
    The associated probability of failure per round is:
    \[\mathbb{P}_{\text{$XX$ routing patch}} =P_L^X+\alpha{\left(\frac{p}{p_{\text{th}}}\right)}^{\frac{d_z+1}{2}}\]
    where $d_z> d$ is the distance between the two $Z$ boundaries of the big rectangular patch.
    In our estimation we do the conservative assumption $d_z=d$.
    $P_L^X$ is given by (\autoref{eq:multiply split log er rate model}):
    \begin{multline*}
    P_L^X \approx \alpha_1n_{\text{seam}}{\left(\frac{p_{\text{Bell}}}{p_{\text{Bell}}^*}\right)}^{\frac{d+1}{2}} + \alpha_{2,\text{tot.\@}}{\left(\frac{p}{p^*}\right)}^{\frac{d+1}{2}} \\
    + \alpha_3 n_{\text{seam}} \sum_{1\leq i\leq d}{\left(\frac{p_{\text{Bell}}}{p_{\text{Bell}}^{**}}\right)}^{\frac{i}{2}} {\left( \frac{p}{p^*}\right)}^{\frac{d+1-i}{2}}
    \end{multline*}
    where $\alpha_{2,\text{tot.\@}}=  \frac{(2d+2)n_{\text{proc.\@}}+(d+1)n_{\text{rows}}}{d}\alpha_2$; the prefactor is the maximal length of the auxiliary patch in units of distance.

    \item We use the regular surface code model \autoref{eq:log error rate model} on a rectangular patch representing the second merging operation during $d$ cycles.
    For this operation the two qubits are always on the same processor justifying the regular surface code logical error rate model.
    Indeed we always choose to put the auxiliary $\ket{0}$ at the bottom of the processor where seats the logical control qubit.
    The associated probability of failure per round is:
    \[\mathbb{P}_{\text{$ZZ$ routing patch}} =\left[\alpha+\alpha \frac{(d+1)n_{\text{rows}}}{d} \right]{\left(\frac{p}{p_{\text{th}}}\right)}^{\frac{d+1}2}\]
    The first $\alpha$ is the $X$ error contribution and the other one is the $Z$ error increased linearly by the height of the rectangular patch.
    This model is conservative as well as in practice the rectangular patch may be shorter than the height of a processor and in any case the distance $d_X$ would be greater than d.
\end{itemize}
Finally, the total logical probability of error during a logical CNOT gate is given by:
\begin{multline*}
    \mathbb{P}^{\text{failure}}_{\text{C}X}
    = 1-{(1-\mathbb{P}_{\text{log.\@ qubits}} )}^{4d} {(1-\mathbb{P}_{\text{$XX$ routing patch}})}^{d}\\
    {(1-\mathbb{P}_{\text{$ZZ$ routing patch}})}^{d}
\end{multline*}
One can notice that the exact same model also works for C$X...X$ gates.
Similarly, C$Z$ gate is performed using an analog procedure implying a $M_{ZX}$ and a $M_{ZZ}$ measurement~\cite{fowler2019lowoverheadquantumcomputation} so that the same noise model is applicable.

The only errors neglected in this model are time-like errors~\cite{domokos2024characterizationerrorscnotsurface}. Such logical errors are chains of errors that link two temporal boundaries of the same type, represented by horizontal tiles on \autoref{fig:spacetime cnot}. Time-like errors don't introduce logical errors in memory experiments.
They can only produce $Z$ logical errors on a $\ket{0}$ memory experiment as the timelike boundaries are of $Z$ type ($Z$ reset and $Z$ measurement). Indeed, in such experiment, the first and last syndrome measurements can only produce reliable $Z$- checks that will detect $X$-errors.
Time-like errors do matter in lattice surgery~\cite{Gidney_2022_stability,Chamberland_2022} as the stored state isn't necessary $\ket{0}$.
However such logical time-like errors scale as $\propto {\left(\frac{p}{p_{\text{mes.\@ th}}}\right)}^{d_m}$ and so can easily be made negligible as compared to topological errors by increasing $d_m$, the number of rounds of stabilizer measurements, resulting in a slight time overhead.
Note that the measurement teleportation on the seam affects both types of logical time-like errors.
Indeed they increase stabilizer measurement error rates for both $X$ and $Z$ stabilizers that are neighbors of the seam.
However during $M_{XX}$ (resp.\@ $M_{ZZ}$) only $Z$ (resp.\@ $X$) time-like errors can happen due to their respective type of temporal boundaries.

Following the way our ansatz has been built, by counting paths in a 3D-lattice, time-like errors could be captured by this .
We expect that simulating a full $M_{XX}$ procedure would allow us to fit our ansatz on such experiments with good fit agreement.
Simulating complete multi-Pauli measurements via lattice-surgery would strengthen the validity of our model and is left as future work.

%% file: appendix_magic_state_v2.tex
\section{Toffoli gate}\label{Annexe: Distillation}
In order to achieve universality, we employ magic state distillation.
This method circumvents the fact that the surface code doesn't provide a direct fault-tolerant constant depth implementation of non-Clifford gates.
The principle is to use noisy magic states and Clifford gates to distill good quality magic states, from which the non-Clifford gate is teleported (using only Clifford gates in addition to the magic state)~\cite{Bravyi_2005}.

The most computationally expensive subroutines of Shor's algorithm are found in modular exponentiation, which relies on adders where the non-Clifford operations are implemented using Toffoli gates~\cite{Gouzien_2021, Gidney_2021}.
For this reason, we will use in our architecture $\ket{\text{CC}Z}$ magic state factories.

The $\ket{\text{CC}Z}$ are injected using the following circuit that we compiled using ZX-calculus from the CNOT-based version~\cite{Gouzien_2023}:

\begin{center}

\scalebox{0.5}{
        \begin{quantikz}
                          &\wireoverride{n}&&\gate[5]{\text{\rotatebox{90}{$M_{ZIIZ}$}}}&        &         &&\gate{X}&&&        &\ctrl{2}&\ctrl{1}& &&\gate{Z}&\\
                          &\wireoverride{n}&&        &\gate[6]{\text{\rotatebox{90}{$M_{ZIIZ}$ }}}&         &&&\gate{X}&&\ctrl{1}&        &\gate{Z}& &\gate{Z}&&\\
                          &\wireoverride{n}&&        &        &\gate[7]{\text{\rotatebox{90}{$M_{XIIX}$ }}} &&&&\gate{Z}&\targ{} & \targ{}&          &\gate{X}&&&\\
	   \lstick{} \setwiretype{n} 		 &&		  &	      &\setwiretype{c}         &             &&&&&\cwbend{-1}\setwiretype{c}&&&\cwbend{-1} &\cwbend{-2}&\setwiretype{n}&\\
        \gategroup[5, steps=3, style={dashed, rounded corners,xshift=-0.5cm}, label style={label position=below, yshift=-0.6cm}]{$\ket{\text{CC}Z}$}
	   \lstick{$\ket{+}$} &\ctrl{4}		 &&		  &        &         &\gate{M_X}   &   \cwbend{-4}     &\setwiretype{n}&&            &.          &               &\\
	   \lstick{} \setwiretype{n} 		 &&		  &	      &         &\setwiretype{c}& &&&&&\cwbend{-3}&&\cwbend{-3}&&\cwbend{-5}&\setwiretype{n}\\
        \lstick{$\ket{+}$}&\control{}      &&        &	      &         &\gate{M_X}   & \setwiretype{c}  &\cwbend{-5}&\setwiretype{n}&            &.          &               &\\
	   \lstick{} \setwiretype{n} 		 &&		  &	      &         &             & \setwiretype{c}&&&&&&\cwbend{-6}&          &\cwbend{-4}&\cwbend{-7}& \setwiretype{n}\\
        \lstick{$\ket{+}$}&\control{}      &&\gate{H}&        &.        &\gate{M_Z}  &  \setwiretype{c}          &&\cwbend{-6}& \setwiretype{n}           &.          &               &
        \end{quantikz}
        }

\end{center}
Pauli corrections are applied in software as an update of the Pauli frame.
As a consequence, the Pauli gate duration doesn't play a role.
However, we still need to wait for the measurement outcome to be classically interpreted before moving on to the next gate for two reasons.
First, because of the C$X$ and C$Z$ corrections that can't easily be made in software.
Secondly, the Pauli corrections must be determined before the next non-Clifford gate as they might switch the measurement results of the next Toffoli teleportation, and change the corresponding corrections.
To overcome the first point, one could use autocorrected magic states~\cite{Litinski_2019, gidney2019flexiblelayoutsurfacecode}.
At the price of adding new auxiliary logical qubits to the magic state and measuring those ancillary qubits on a basis that depends on the previous measurements, only Pauli corrections are required to perform Toffoli gates, and the following Clifford gates can be applied without waiting to know those corrections.
In our case, without parallelization, the only benefit of using Auto-$\text{CC}Z$ magic state would be to reduce the duration of the injection scheme from $\underbrace{7d}_{\text{interaction}}+1.5\underbrace{(4d)}_{\text{CNOT corrections}}=13d$ to $7d$.
It would require a higher number of auxiliary states and longer state preparation, and so more factories. Hence, we choose to use regular $\text{CC}Z$ magic states.
The main benefit of using Auto-CCZ state (and the reason why~\cite{gidney2019flexiblelayoutsurfacecode} introduce it) is that as Pauli corrections are applied in software one can parallelize the Toffoli gates (One doesn’t need to wait for Clifford corrections to be performed before starting the next Toffoli injection (see figure 7 of~\cite{gidney2019flexiblelayoutsurfacecode}). As we use no parallelization (because of the congestion it would induce in a minimalist chain like architecture with only one link between processors) we would not be able to benefit from this.

\subsection{Magic state factories}\label{annexe sec:Magic state factories}

Factoring a 2048-bit RSA number typically requires on the order of ${10}^{10}$ Toffoli gates.
As a consequence, one needs to have $\ket{\text{CC}Z}$ state with an error rate of at most ${10}^{-11}$ in order to ensure the reliability of the algorithm.
Otherwise, the expected error rate of the whole algorithm would explode.

To produce good quality $\ket{\text{CC}Z}$ states, we use the synthillation two-level protocol described in~\cite{Litinski_2019_magic}.
This two-level protocol works as follows: it uses $n_{l1}$ 15-to-1 $T$ state factory that each produces states of quality up to $35{p'}^3$ where $p'$ is the error rate of a $\ket{T}$ on a physical qubit.
Those $\ket{T}$ states are then used to perform the non-Clifford gates of a 8$T$-to-$\text{CC}Z$ protocol reducing the best achievable error rate to $28{(35{p'}^3)}^2$.\footnote{Note that under our depolarizing noise model of paramater $p=\num{1e-3}$, the two level protocol saturate at $28{(35\left(\frac{2p}{3}\right)^3)}^2\approx\num{3e-15}$ because $X$ errors on $\ket{T}$ states are \SI{50}{\percent} of the time harmless and \SI{50}{\percent} of the time equivalent to a $Z$ errors. See~\cite{Litinski_2019_magic} for details and comparison to other noise models.}

Using the code provided in the appendix of~\cite{Litinski_2019_magic}, we have been able using our circuit-level threshold $p_{\text{th}} = 7.5 \times {10}^{-3}$ to build a set of $\ket{\text{CC}Z}$ factories by sweeping over a large set of the parameters $d_X,d_Z,d_m,d_{X2},d_{Z2},d_{m2}, n_{l1}$, the distances of the logical qubits in each level described before.
Indeed, in~\cite{Litinski_2019_magic}, Litinski showed that one can use different $X$ and $Z$ distances for check qubits in a distillation protocols as well as a shorter timelike distance (number of round) during the lattice surgery operations within the factory.
Here, $d_X,d_Z,d_m$ are the $X$, $Z$ and time-like distances of the first $15T\rightarrow 1T$ level of factories and $d_{X2},d_{Z2},d_{m2}$ are the one of the second level, namely the $8T\rightarrow CCZ$ protocol.
The layout of such two-level factory proposed by Litinski and used in our architecture is available in \autoref{fig:factory layout}.
Note that the qubit storing the output magic is can't have a reduce $Z$ distance, thus it is a square patch whose width is given by the $X$ distance parameter.

\begin{figure}
    \centering
    \includegraphics[width=\linewidth]{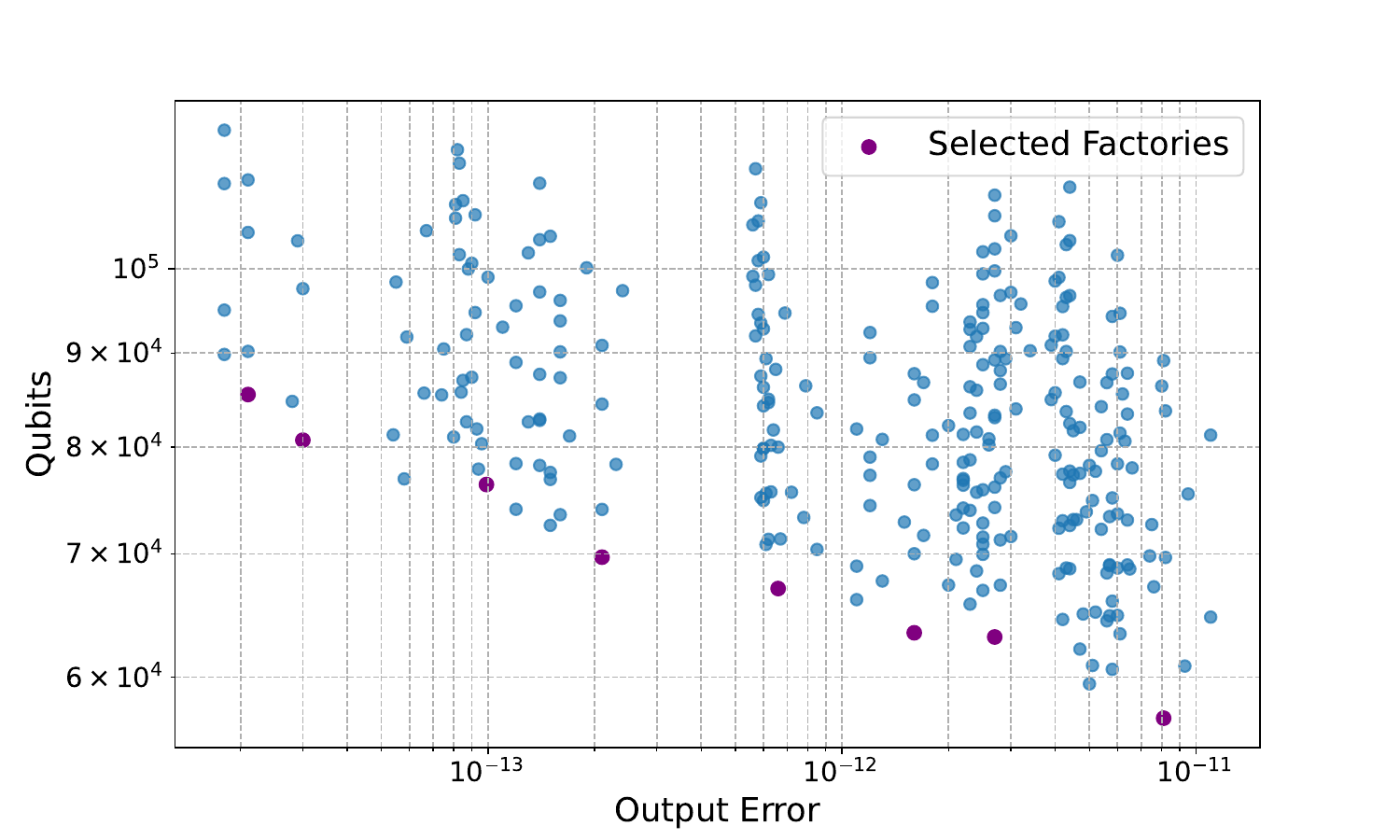}
    \caption{Two level magic state factories $15T\to 1T$ into $8T\to 1\text{CC}Z$ to produce $\ket{CCZ}$ magic states from~\cite{Litinski_2019_magic}.
    Each point is a factory with different parameters (code distances and number of first stage factories).
    We selected a set of ``good'' factories colored in purple on the figure.
    They minimize the qubit footprint for given output error rates.
    In practice, most of the time the optimization algorithm choose the factory on the very left of the figure.}\label{fig:factories}
\end{figure}

Among this cohort of factories displayed in \autoref{fig:factories} we select a few of them, the one with the smallest qubit footprints.
Then our resource estimation run with all of these factories and selects the one that provides the best performance for running Shor algorithm.
The selected factory also determines the size of the processors in our architecture.
We didn't dive into the subject of distributing such factories.
However, two-level protocols may be particularly handy to do so as it requires a logical routing connection with the level two factory only once, to use the output $\ket{T}$ state via lattice surgery.
We could imagine a central chip hosting the 4 logical qubits of the 8-to-$\text{CC}Z$ connected via Bell pairs to $n_{l1}$ different smaller chips hosting 15-to-1 $\ket{T}$ factories.
We gather all the information of the different factories in \autoref{tab: factories}.

\begin{table*}[t]
\centering
\renewcommand{\arraystretch}{1.2}
\begin{tabular}{|l|c|c|c|c|c|c|c|c|c|}
\hline
  & Factory 1 & Factory 2 & Factory 3 & Factory 4 & Factory 5 & Factory 6 & Factory 7 & Factory 8  \\ 
\hline
Output error & $2.1 \times {10}^{-14}$ & $3.0 \times {10}^{-14}$ & $9.9 \times {10}^{-14}$ & $2.1 \times {10}^{-13}$ & $6.6 \times {10}^{-13}$ & $1.6 \times {10}^{-12}$ & $2.7 \times {10}^{-12}$ & $8.1 \times {10}^{-12}$  \\ 
\hline
Qubitcycles  & $\num{11394367}$ & $\num{8823434}$ & $\num{8347167}$ & $\num{7620887}$ & $\num{7327490}$ & $\num{6931679}$ & $\num{6896807}$ & $\num{6229603}$  \\ 
\hline
Qubits       & $\num{85436}$ & $\num{80700}$ & $\num{76344}$ & $\num{69716}$ & $\num{67032}$ & $\num{63424}$ & $\num{63092}$ & $\num{57000}$  \\ 
\hline
Codecycles   & $133.4$ & $109.3$ & $109.3$ & $109.3$ & $109.3$ & $109.3$ & $109.3$ & $109.3$  \\ 
\hline
dx           & $21$ & $21$ & $21$ & $19$ & $19$ & $17$ & $19$ & $17$  \\ 
\hline
dz           & $11$ & $9$ & $9$ & $9$ & $9$ & $9$ & $9$ & $9$ \\ 
\hline
dm           & $11$ & $9$ & $9$ & $9$ & $9$ & $9$ & $9$ & $9$  \\ 
\hline
dx2          & $35$ & $35$ & $35$ & $33$ & $31$ & $31$ & $31$ & $29$ \\ 
\hline
dz2          & $21$ & $21$ & $19$ & $19$ & $19$ & $19$ & $17$ & $17$ \\ 
\hline
dm2          & $21$ & $21$ & $19$ & $19$ & $19$ & $19$ & $17$ & $17$ \\ 
\hline
nl1          & $4$ & $4$ & $4$ & $4$ & $4$ & $4$ & $4$ & $4$\\ 
\hline
\end{tabular}
\caption{Summary of two-level synthillation factories parameters generated from the supplementary material of~\cite{Litinski_2019_magic}.}
\label{tab: factories}
\end{table*}

On top of the two level factory layout considered by Litinski, we add an extension, see~\autoref{fig:factory layout}, to one of the two routing areas in order to ensure that the magic states are injected with a distance $d$ patch during lattice surgery, especially near the processor frontiers.
Indeed, it can happen for high $p_{\text{Bell}}$ that $d_{X2}<d$, see~\autoref{tab:resource}.

\begin{figure}
    \centering
    \includegraphics[width=0.9\linewidth]{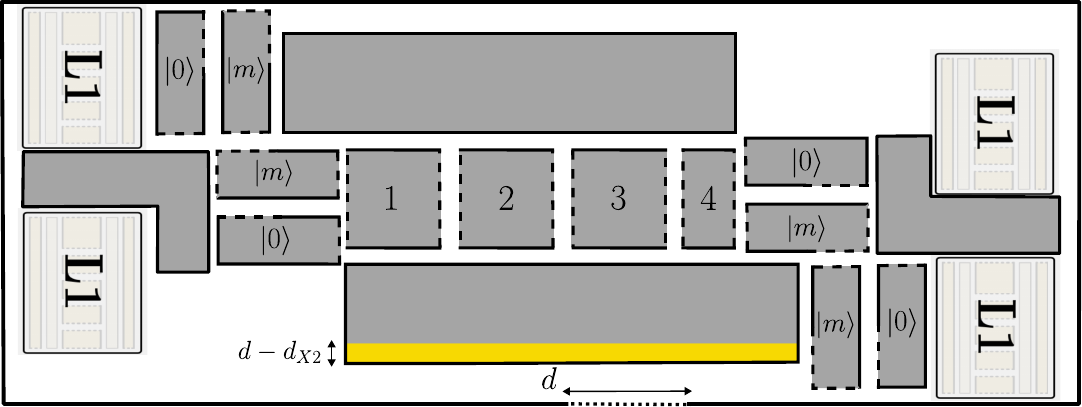}
    \caption{
    $15T\rightarrow1T$ into $8T\rightarrow \text{CC}Z$ Two level factory layout proposed in~\cite{Litinski_2019_magic}.
    The yellow region is an extension to the routing region allowing to use a distance $d$ patch in lattice surgery.
    $d$ is the distance of the logical qubits within the processors (chosen by the resource estimator) while $d_{X2}$ is the distance of the surface code patch within the second stage of distillation.
    Note that in this layout the qubits 1,2,3 are the 3 qubits storing the $\ket{\text{CC}Z}$ state while the qubit 4 is the check qubit of the $8T\rightarrow \text{CC}Z$ protocol (see Fig.16b and section 5 of~\cite{Litinski_2019_magic} for more details).
    The routing areas above and below those qubits are used to perform the rotations using multi-qubit $Z$ measurement via state injection.
    The layout only give access to the $Z$ logical operator of qubits 1,2,3 as the distillation protocol only requires joint $Z$ measurement.
    $L1$ boxes represent the first-level $15T\mapsto 1T$ factories.
    $\ket{m}$ boxes stores the magic state outputs of the first-level factories.
    }
    
    \label{fig:factory layout}
\end{figure}

\subsection{Logical error model of a Toffoli}\label{annexe sec: Toffoli error model}
The cost of a Toffoli gate in our resource estimation is separated in three parts: 
\begin{enumerate}
    \item The cost of $\ket{\text{CC}Z}$ factories, detailed in \autoref{annexe sec:Magic state factories}, is deduced from its probability to fail and its duration (combining duration of one shot and heralding probability) in cycle times.
    The probability for the $\ket{\text{CC}Z}$ synthillation protocol to fail that we note $\mathbb{P}_{\text{factory}}$ is computed thanks to the code provided in~\cite{Litinski_2019_magic}.
    \item The interaction between the logical qubits over which we want to perform the Toffoli and the magic state.
    This step implies 3 lattice surgery operations and one Hadamard as shown in the injection circuit above.
    Note that accessing the logical Z operator of the logical qubits within a processor requires $d$ times step to move it into the central routing corridor (see~\autoref{fig: compact layout}).
    Therefore the two first $M_{ZZ}$ operations each requires $3d$ timesteps.
    Note however that the $M_{XX}$ operation requires only $d$ rounds as the transversal Hadamard effectively give provide access to the $X$ logical operator of the third patch via the routing area.
    No additional noise is taken for the Hadamard itself as it has a logical error rate  comparable to a quantum memory experiment~\cite{Geh_r_2024}. Logical idling noise on every logical qubits which we take into account in the CNOT thus model the noise of the Hadamard gate.  
    Hence, the logical error rate of the interaction part is modeled as: 
    \[\mathbb{P}_{\text{interact.\@}} = 1 - {\left(1-\mathbb{P}_{\text{log.\@ qubits}}\right)}^{7d} \left(1- \mathbb{P}_{\text{routing patch}} \right)^{3d} \]
    where 
    \begin{multline*}
        \mathbb{P}_{\text{routing patch}} = \alpha_1n_{\text{seam}}{\left(\frac{p_{\text{Bell}}}{p_{\text{Bell}}^*}\right)}^{\frac{d+1}{2}} \\+  \alpha_{2,\text{tot.\@}}{\left(\frac{p}{p^*}\right)}^{\frac{d+1}{2}} 
    \\+ \alpha_3 n_{\text{seam}} \sum_{1\leq i\leq d}{\left(\frac{p_{\text{Bell}}}{p_{\text{Bell}}^{**}}\right)}^{\frac{i}{2}} {\left( \frac{p}{p^*}\right)}^{\frac{d+1-i}{2}}
    \end{multline*}
    with $\alpha_{2,\text{tot.\@}}=  \frac{(2d+2)n_{\text{proc.\@}}+(d+1)n_{\text{rows}}}{d}\alpha_2$,
    models the noise of the routing patch accross the chain of processors
    and \[\mathbb{P}_{\text{log.\@ qubits}} =1-{\left(1-2\alpha{\left(\frac{p}{p_{\text{th}}}\right)}^{\frac{d+1}{2}}\right)}^{N_{\text{log.}}}\]
    models the noise of idling logical qubits in the meantime.
    \item The fixing step which involves 3 measurements that are performed in parallel and 3 eventual Clifford corrections.
    The time it takes to perform a measurement (and the corresponding classical control) is called $t_r$ and the noise of the 3 measurements is modeled as idling errors over all logical qubits during this time $t_r$.
    Concerning the corrections, each of them is required only \SI{50}{\percent} of the time as they are controlled by fully random measurement outcomes.
    Adding the time and error cost of the 1.5 expected logical CNOT (or C$Z$) gates, we get the overall logical error rate of the fixing step: 
    \[
    \mathbb{P}_{\text{fixing}}=1-{\left(1-\mathbb{P}_{\text{log.\@ qubits}}\right)}^{\frac{t_r}{t_c}} {\left(1-\mathbb{P}_{\text{C}X}^{\text{failure}}\right)}^{1.5}.
    \]
\end{enumerate}

The overall failure probability of a Toffoli gate is thus given by: 
\[
\mathbb{P}_{\text{CC}X}= 1-\left(1-\mathbb{P}_{\text{factory}}\right)\left(1-\mathbb{P}_{\text{interact.\@}}\right)\left(1-\mathbb{P}_{\text{fixing}}\right).
\]

Note that the time cost of one $\text{CC}Z$ state preparation is typically on the order of hundreds of cycles which is less than the duration of a Toffoli injection as for $d=30$, $7d =210$.
Therefore we use two factories on the architecture so that the next $\text{CC}Z$ can always be produced while injecting the one from the other factory.
The time cost of a Toffoli gate is thus given by the injection time, i.\@e $T_{\text{CC}X}= 7d+ 1.5T_{\text{C}X}+t_r$ where $T_{\text{C}X}=4dt_c$ for the compact layout.
Note that using Auto-$\text{CC}Z$ magic state would only reduce the duration of a Toffoli gate to $ 7d+t_r$ (as one can't achieve reaction-limited computation without parallelization).

%% file: appendix_shor.tex
\section{Resource estimation for Shor's algorithm}

\subsection{Details on Shor's algorithm execution}
\label{End_matter:resource}

Ekerå and Håstad's version of Shor's Algorithm for factoring an $n$-bit integer $N$ starts by picking a random integer $g$ in $\mathbb{Z}_N^*$, then uses modular exponentiation over a quantum register before applying a quantum Fourier transform operation.
The cost of the modular exponentiation operation, namely $\ket{e}\ket{1} \mapsto \ket{e} \ket{g^e \text{ mod }N}$, dominates the overall cost of factorization~\cite{Gidney_2021,Gouzien_2021}; therefore, our resource estimation focuses on estimating the cost of this modular exponentiation.

The modular exponentiation operation is decomposed into product operations as $g^e=\prod_{i=0}^{n_e-1}{g^{2^ie_i}}$, where $e_0...e_{n_e-1}$ is the binary decomposition of $e$.
Each product $\ket{e_i}\ket{a}\ket{0} \mapsto \ket{e_i}|ag^{2^ie_i}\rangle  \ket{0}$ (where $a$ is the partial product $a=\prod_{j=0}^{i-1}{g^{2^j}}^{e_j}$) is then decomposed as:
\begin{equation*}
 \ket{e_i}\ket{a}\ket{0} \xrightarrow{1} \ket{e_i}\ket{a} |\underbrace{0+ag^{2^i}e_i}_{\equiv b}\rangle \xrightarrow{2}  \ket{e_i}|a- bg^{-2^i}e_i\rangle \ket{b}
\end{equation*}
where step $1$ and step $2$ are product additions controlled by the value of $e_i$. 
If $e_i=0$, the second and third registers encode $\prod_{j=0}^{i}{g^{2^je_j}}$ (as ${g^{2^ie_i}}=1$) and $0$ whereas if $e_i=1$, they encode $0$ and $\prod_{j=0}^{i}{g^{2^je_j}}$ respectively. By swapping the two last registers if $e_i=1$ and by repeating the two product additions for each $i$ from $0$ to $n_e-1$, the second register ends up encoding $g^e$ and the last one returns to $0$.
In practice, product-additions are decomposed as a series of doubly-controlled additions.
For example, starting from the binary decomposition $a_n \ldots a_0$ of $a$, step $1$ is implemented as a series of addition of $2^jg^{2^i}$ if and only if $e_i=a_j=1$.

In our resource estimation, we use Windowed arithmetic techniques enabling controlled additions to be performed by blocks see~\cite{gidney2019windowedquantumarithmetic} and Fig.~9 of~\cite{Gouzien_2021}.
The decomposition that we use is now $g^e=\prod_{i=0 \text{ mod } w_e}^{n_e-1}{g^{2^ie_{i:i+w_e}}}$ (where $x_{\ell:m}=\sum_{k=\ell}^{m-1}2^{k-\ell}x_k$).
We start with an auxiliary register initially in $a^{[0]}=1$ that after the $i$-th product must end up in state
$a^{[i]}=\prod_{k=0 \text{ mod } w_e}^{i-1}{g^{2^ke_{k:k+w_e}}}$.
Each product is achieved by reading respectively $w_e$ and $w_m$ qubits from $e$ and $a$ registers, loading the value $T_{i,j}\equiv2^j a_{j:j+w_m}^{[i]}g^{2^i e_{i:i+w_e}}$ in an extra auxiliary register of size $n$.
This temporary register is then used as input of a modulo $N$ addition, completing the following step $1'$ 
\begin{multline*}
\ket{e_{i:i+w_e}}\ket{a_{j:j+w_m}^{[i]}}\ket{T_{i,j}}\ket{c_{i,j}} \xrightarrow{1'} \\
\ket{e_{i:i+w_e}}\ket{a_{j:j+w_m}^{[i]}} \ket{T_{i,j}}|c_{i,j}+T_{i,j}\rangle 
\end{multline*}
with $c_{i,j}= \sum_{k<\lfloor j/w_m\rfloor}T_{i,k}$ and $c_{i,0}=0$.
Afterwards $T_{i,j}$ is properly unloaded from the auxiliary register.
This three steps process (load, addition, unload) is repeated for all $j$ values that are multiples of $w_m$ ending up with registers in state $\ket{e_{i:i+w_e}}\ket{a^{[i]}} \ket{0}\ket{a^{[i+w_e]}}$ and hence completing $w_e$ step $1$ at once.

Then, we do the same for the second product addition, analogously to $2$, in order to clean the second register before swapping it with the last one. 
This block is repeated for all $i$ multiples of $w_e$.
This trick reduces the number of multiplications by a factor $w_e$ and the number of additions in each of them by a factor $w_m$.
The SWAPs are no longer controlled, making them easier to do.

Following the coset representation of integers~\cite{Gidney2019Approximateencodedpermutations}, the addition of $y$ modulo N to a quantum register encoding the integer $x$ is approximately realized with a regular adder by encoding $x$ as
\[\ket{\text{Coset}_{c_{\text{pad}}}(x)}= \frac{1}{\sqrt{2^{c_{\text{pad}}}}}\sum_{j=0}^{2^{c_{\text{pad}}}-1}\ket{x+jN}\]
with $c_{\text{pad}}$ additional qubits. The error on the result, taken into account in our estimation, is exponentially reduced with $c_{\text{pad}}$.

\begin{figure*}
    \centering
    \includegraphics[width=0.8\linewidth]{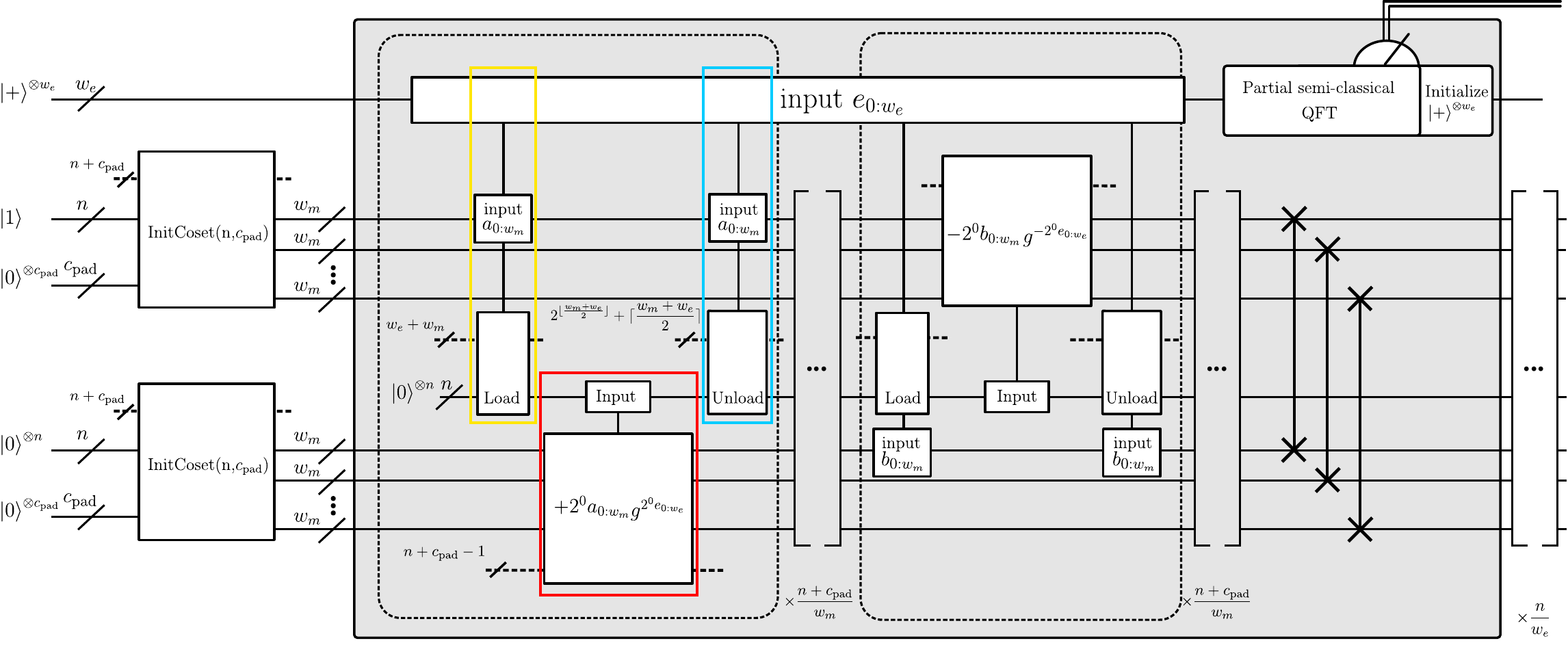}
    \caption{
    Full circuit of the factorization algorithm. Full lines are registers whose sizes are indicated by crossing segments. Dashed lines represents ancillary registers used temporarily in the circuit.
    The total number of logical qubit is $4n+3c_{\text{pad}}+w_e-1$ corresponds to the maximal amount of simultaneously used qubits.
    The algorithm starts with two coset register initializations.
    The grey box, repeated $\frac{n_e}{w_e}$ times, performs $\ket{e_{i:i+w_e}}|a^{[i]}\rangle \ket{0}\ket{0}\mapsto \ket{e_{i:i+w_e}}|a^{[i+w_e]}\rangle \ket{0}\ket{0}$. It contains a first circuit (boxed in dashed line) performing 
    $\ket{e_{i:i+w_e}}|a^{[i]}\rangle \ket{0}\ket{c_{i,j}}\mapsto \ket{e_{i:i+w_e}}|a^{[i]}\rangle \ket{0}\ket{c_{i,j}+T_{i,j}}$ with a lookup (yellow), an addition (red) and an unlookup (blue). It is repeated $(n+c_{\text{pad}})/w_m$ times to end up in $\ket{e_{i:i+w_e}}|a^{[i]}\rangle \ket{0}|a^{[i+w_e]}\rangle$.
   A second circuit (boxed in dashed line), repeated $(n+c_{\text{pad}})/w_m$ times, puts the second register in $\ket{0}$.
    The SWAP gates exchanges the second and last register. Finally a semi-classical QFT is performed.
    The first register of size $w_e$ stores a part of $e$ while the second register stores $g^e$ at the end of the computation. 
    }
    \label{fig:full_algo}
\end{figure*}

An overview of the subroutines is provided in \autoref{fig:full_algo}, which specifies the required number of logical qubits. For the gate count, note first that the cost of the semi-classical QFT and the SWAPs is negligible compared to the cost of the exponentiation, see Appendix A of~\cite{Gouzien_2021}.
Hence, we provide below the gate count of each relevant subroutines, namely the initialization of the coset representation, the lookup operation for loading $T_{i,j}$, the adder and the unlookup operation for unloading $T_{i,j}$.

The coset register initialization circuit takes
    \begin{multline*} c_{\text{pad}}(1.5(n+c_{\text{pad}})-1) (\text{Init.}+\text{Meas.}) \\+
    c_{\text{pad}}(6(n+c_{\text{pad}})-8)\text{CNOT} \\+
    c_{\text{pad}}(1.5(n+c_{\text{pad}})-2.5)\text{CC}X
    \end{multline*}
    operations where $\text{Init.}$ is a logical $\ket{0}$ or $\ket{+}$ preparation and $\text{Meas.}$ is a logical measurement in $Z$ or $X$ basis, see Fig.~7 and Fig.~8 of~\cite{Gouzien_2021}.

The lookup table circuit takes (see Fig.~10 of~\cite{Gouzien_2021}) 
    \begin{multline*}
        2(2^{w_m+w_e}-2)\text{CNOT} + 2^{w_m+w_e}\text{C}X\ldots X \\+(2^{w_m+w_e}-2)(\text{CC}X+\text{Init.}+\text{Meas.})
    \end{multline*}
    operations. Given the native lattice surgery parallelization for C$X\ldots X$ gates, we consider them as costly as a CNOT gate (see \autoref{fig:cnotnot circuit} of \autoref{Annexe lattice surgery}).

The cost of the adder (see~\cite{Gidney_2018_adder} and Fig.~13 of \cite{Gouzien_2021}) is
    \begin{multline*}
        (5.5(n+c_{\text{pad}})-7.5)\text{CNOT}\\+
        (n+c_{\text{pad}}-1)(\text{CC}X+\text{Init.}+\text{Meas.}).
    \end{multline*}

The cost of the Unlookup circuit (see appendix D.3 of~\cite{Gouzien_2021}) is
    \begin{multline*}
         (5\times2^{\ceil{\frac{we+wm}{2}}-1} + 5\times2^{\floor{\frac{we+wm}{2}}-1}- 5.5)\text{CNOT}\\+
         (2^{\ceil{\frac{we+wm}{2}}} + 2^{\floor{\frac{we+wm}{2}}}- 3)\text{CC}X\\+
         (2^{\ceil{\frac{we+wm}{2}}} + 2^{\floor{\frac{we+wm}{2}}}- 2)\text{Init.}\\+
         (n+2^{\ceil{\frac{we+wm}{2}}} + 2^{\floor{\frac{we+wm}{2}}}- 3)\text{Meas.}).
    \end{multline*}
Note that the Hadamard appearing in Fig.~11 of~\cite{Gouzien_2021} has been merged with the multi-target CNOT of the lookup circuit, turning them into C$Z\ldots Z$.

Further note that the initializations of auxiliary qubits are done upstream so that we don't wait for them during the computation.
Most measurements can't be done offline as their outcomes are used as classical controls of gates. Logical gates are applied sequentially.

The code used for memory experiment simulations, fits of error model and resource estimation is available here~\footnote{\url{https://github.com/HugoJ1/shor_multi_chip}}.

\subsection{Reading guide of the resource estimation table}\label{Reading guide table}
The resource estimation has been built in the same manner as in~\cite{Gouzien_2021,Gouzien_2023}; for each $(p, p_{\text{Bell}})$, the algorithm performs the resource estimation for different distances, magic-state factories, sizes of windowed arithmetic~\cite{gidney2019windowedquantumarithmetic} and coset representation of integers~\cite{Gidney2019Approximateencodedpermutations} size.
For each set of parameters, the duration $\Delta t$ of one run of the algorithm as well as its probability to fail $p_{\text{fail}}$ are computed.
$p_{\text{fail}}$ takes into account the probability that the classical post-routine fails as well as the probability that a topological error happens at any point in the computation.
Each execution of Shor algorithm can be viewed as a trial of a Bernoulli process of parameter $p_{\text{fail}}$.
The expected number of trials to factorize a 2048-RSA integer is then $\frac{1}{p_{\text{fail}}}$.
So we can define the expected duration to factorize a 2048-RSA integer as $\Tilde{\Delta} t = \frac{\Delta t}{p_{\text{fail}}}$.
Then the best set of parameters is the one that minimizes the metric $\Tilde{\Delta} t n_{\text{phys.\@}}$, where $ n_{\text{phys.\@}}$ is the total number of physical qubits in the architecture.

In all of our estimation the physical error rate $p$ is fixed at $p=\num{1e-3}$.
To evaluate the overhead cost due to Bell pair infidelity in our distributed architecture, \autoref{fig:curve resource estimation} shows the two curves of respectively $\Tilde{\Delta}t$ in blue and $ n_{\text{phys.\@}}$ in green for values of $p_{\text{Bell}}$ between $0$ and $0.05$.
This range of Bell pair noise has been chosen to ensure the validity of our model which is satisfied for $p_{\text{Bell}} < \SI{5}{\percent}$.
As expected, when $p_{\text{Bell}} \to 0$ both curves are flat, as the logical error rate is dominated by the physical error rate $p$ contribution which manifests as plateaus on the logical error rate in \autoref{fig:p fixed slice}.
Qubits overhead $n_{\text{phys.\@}}$ is mostly sensitive to the distance $d$ of the surface code patches whereas $\Tilde{\Delta} t$ is sensitive to any, even small, increase in the logical error rate, that is derived from the logical error rate model, as we consider the average duration of the algorithm, knowing that the result can be checked and the algorithm rerun in case of failure.
As can be seen on \autoref{fig:curve resource estimation}, at each raise of the distance $d$ we observe an expected increase of $n_{\text{phys.\@}}$ as well as a decrease of $\Tilde{\Delta} t$ due to the lowering of the logical error rate.
The first increase of the distance happens around $p_{\text{Bell}} = \SI{2}{\percent}$  which means that below this error rate, the only consequence of Bell pair noise is a slight time overhead (of around \SI{5}{\percent}).
Note that in \autoref{tab:resource} the row ``\# of processors'' contains explicit sum of the form $a+b$, where $a$ and $b$ are positive integers. The goal is to inform the reader about processors storing logical qubits involved in the computation (first contribution $a$) and processors storing magic state factories (here $b=2$ in the distributed setting).
The ``\# of physical qubits per processor'' corresponds to the number of physical qubits in a processor of type $a$.
The number of physical qubits in processors of type $b$ can be deduced from the Magic state factory row and \autoref{tab: factories}.
The row ``Total \# of physical qubits'' is the sum of the number of physical qubits in processors of type $a$ and $b$.

To be able to compare those estimations to a base case, we also did the resource estimations for a monolithic architecture with the same compact layout as \autoref{fig: compact layout}, but on a single chip.
In such a monolithic architecture, there are no Bell pairs and so only the regular surface code logical error rate model \autoref{eq:log error rate model} is used.
Results are reported in \autoref{tab:resource}.

\begin{figure}
    \centering
    \includegraphics[width=\linewidth]{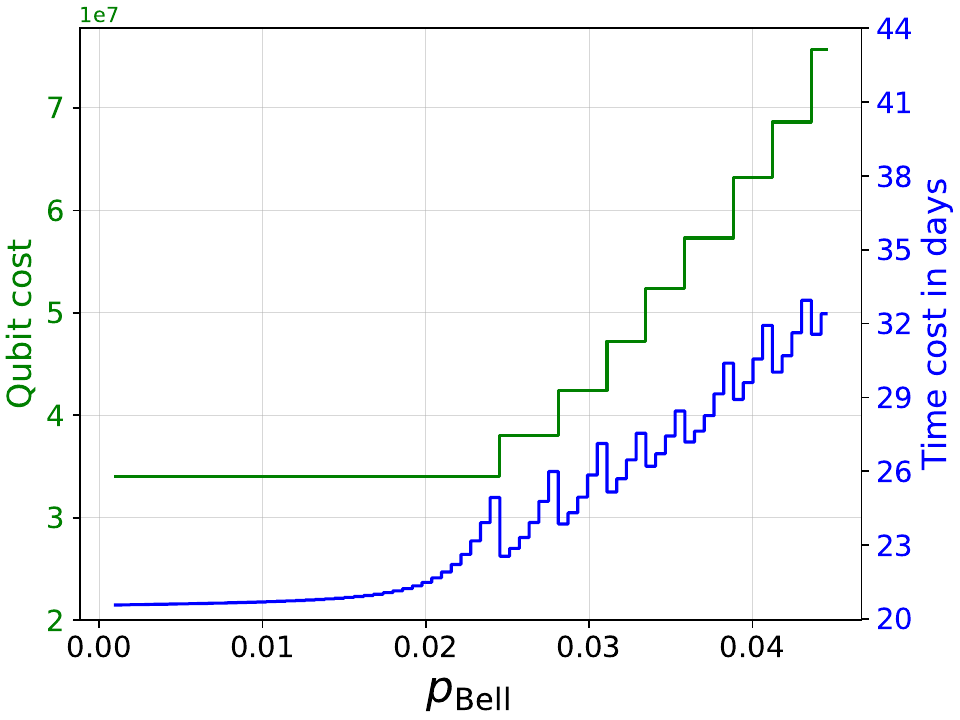}
    \caption{Resource estimation for factoring a 2048-bit RSA number using Shor algorithm in a distributed architecture for different values of interprocessor noise strength.
    Green curve is the number (in tens of millions) of physical qubits required in the architecture.
    Blue curve is the expected duration of the full algorithm, rescaled by the overall the success probability.}
    \label{fig:curve resource estimation}
\end{figure}

\subsection{Resource estimation table}
\onecolumngrid

\begin{table*}
\centering
\begin{tabular}{|c|c|c|c|c|c|}
\hline   
& \multicolumn{1}{|c|}{monolithic} & \multicolumn{4}{|c|}{distributed}\\ 
\cline{2-6}
  & $p=\SI{0.1}{\percent}$ & 
  $p=\SI{0.1}{\percent}$ &
  $p=\SI{0.1}{\percent}$ &  
  $p=\SI{0.1}{\percent}$ &
  $p=\SI{0.1}{\percent}$ \\
                                   &                               &$p_{\text{Bell}}=\SI{0}{\percent}$&$p_{\text{Bell}}=\SI{2}{\percent}$&$p_{\text{Bell}}=\SI{3}{\percent}$& $p_{\text{Bell}}=\SI{4}{\percent}$ \\
\hline
distance $d$                       & 35                              &35                              &35                              & 39                           & 47 \\
\hline
$\#$ logical qubits per processors & 8280                           &22                              &22                              & 18                           & 12 \\
\hline
$\#$ physical qubits per processor &$\num{32040090}$                & $\num{89781}$                  &$\num{89781}$                   &$\num{91917}$                 & $\num{91389}$ \\
\hline
$\#$ processors                    & 1                               & 377+2                          & 377+2                          & 460+2                        & 690+2  \\
\hline
total $\#$ of physical qubits      & $\num{32040090}$                & $\num{34018414}$               & $\num{34018414}$             &$\num{42455161}$            & $\num{63218103}$ \\
\hline
Estimated duration                 & 20 days 9 hours                & 20 days 13 hours               & 21 days 12 hours                & 25 days 20 hours              & 30 days 13 hours\\
\hline
Magic state factory                & $p_{\text{out}}= \num{2.1e-13}$ &$p_{\text{out}}= \num{2.1e-14}$ & $p_{\text{out}}= \num{2.1e-14}$& $p_{\text{out}}= \num{2.1e-14}$&$p_{\text{out}}= \num{9.9e-14}$ \\
\hline
Space overhead                     & N.A                             & reference                      & $\SI{0}{\percent}$             &$\SI{25}{\percent}$           &$\SI{86}{\percent}$ \\
\hline
Time overhead                      & N.A                             & reference                      & $\SI{5}{\percent}$            &$\SI{26}{\percent}$           &$\SI{49}{\percent}$ \\
\hline
\end{tabular}
\caption{Resource estimations table.
In the distributed setting, the number of processors is written as: $\#$ processors storing logical qubits for computation $+$ $\#$ processors storing a magic state factory.
The number of physical qubits per processor announced in the table is the one of processors storing logical qubits.
For more information about the magic state factories, see \autoref{tab: factories}.
A reading guide of this table is available in \autoref{Reading guide table}.}\label{tab:resource}
\end{table*}
\twocolumngrid